\documentclass[12pt,draftcls,onecolumn]{IEEEtran}

\usepackage{amsmath,amssymb,url,listings,mathrsfs,wasysym}
\usepackage{enumerate}
\usepackage{graphicx}
\usepackage{subfigure}
\usepackage{cite}
\usepackage[linesnumbered,vlined,ruled]{algorithm2e}

\newcommand{\num}[1]{\relax\ifmmode \mathbb #1\else $\mathbb #1$\fi}
\newcommand{\naturals}{{\num N}}
\newcommand{\reals}{{\num R}}
\newcommand{\expect}{\mathbb{E}}
\newcommand{\conv}{\mathit{Conv}}

\newcommand{\controller}{\mathcal{C}}

\newcommand{\J}{{\mathcal{J}}}
\newcommand{\R}{{\mathcal{L}}}
\newcommand{\M}{{\mathcal{M}}}
\newcommand{\N}{{\mathcal{N}}}
\renewcommand{\P}{{\mathcal{P}}}

\newcommand{\Z}{{\mathcal{Z}}}

\newcommand{\pr}{\textit{PR}}
\newcommand{\splitplans}{\mathbb{S}}

\newtheorem{theorem}{Theorem}
\newtheorem{lemma}{Lemma}
\newtheorem{proposition}{Proposition}
\newtheorem{corollary}{Corollary}
\newtheorem{definition}{Definition}

\title{\Large \bf
Throughput Optimal Distributed Traffic Signal Control
}
\author{Tichakorn~Wongpiromsarn,
Tawit~Uthaicharoenpong,
Emilio~Frazzoli,\\
Yu Wang and Danwei Wang
\thanks{T. Wongpiromsarn is with the Ministry of Science and Technology,
Thailand {\tt\footnotesize tichakorn@gmail.com}}
\thanks{T. Uthaicharoenpong is with the Singapore-MIT Alliance for Research and Technology,
Singapore {\tt\footnotesize tawit@smart.mit.edu}}
\thanks{E.  Frazzoli is with the Massachusetts Institute of Technology, Cambridge, MA, USA
  {\tt\footnotesize frazzoli@mit.edu}}
\thanks{Y. Wang and D. Wang are with the Nanyang Technological University, Singapore
{\tt \footnotesize wang0400@e.ntu.edu.sg, edwwang@ntu.edu.sg}}%
}
\begin{document}
\maketitle

\begin{abstract}
We propose a distributed algorithm for controlling traffic signals, allowing constraints such as
periodic switching sequences of phases and minimum and maximum green time to be incorporated.
Our algorithm is adapted from backpressure routing, which has been mainly applied to communication
and power networks.
We formally prove that our algorithm ensures global optimality as it leads to maximum network throughput even though
the controller is constructed and implemented in a completely distributed manner.
\end{abstract}

\section{Introduction}
%Traffic congestion arises in many urban areas because of the continual growth in motorization and 
%the difficulties in increasing road capacity due to space limitations and budget constraints.
%This problem causes significant social and economic losses and leads to excessive air pollution.
%As a result, traffic management that aims at maximizing the efficiency of road networks without 
%increasing road capacity becomes increasingly crucial.

Traffic signal control %provides an emerging application of cyber-physical systems in transportation. 
is an example of cyber-physical systems, featuring tight interaction of
the cyber components and the physical aspects of the system.
It is an important element in traffic management that affects the efficiency of urban transportation.
With growing traffic congestion in many urban areas, 
improving traffic signal control to optimize the efficiency of road networks becomes increasingly important.
Recently, adaptive traffic signal control systems has been employed by many major cities.
These systems improve the performance of traffic management by dynamically adjusting the light timing based on the current traffic situation.
Examples of widely-used adaptive traffic signal control systems include SCATS (Sydney Coordinated Adaptive Traffic System)~\cite{Lowrie82SCATS,Keong93Glide,Liu03Master} and 
SCOOT (Split Cycle Offset Optimisation Technique)~\cite{Day1998,Stevanovic2008}.

Control variables in traffic signal control systems typically include phase, cycle length, split plan and offset.
A phase specifies a combination of one or more traffic movements simultaneously receiving the right of way during a signal interval.
Cycle length is the time required for one complete cycle of signal intervals.
A split plan defines the percentage of the cycle length allocated to each of the phases during a signal cycle.
Offset is used in coordinated traffic control systems to reduce frequent stops at a sequence of junctions.

As explained in \cite{Lowrie82SCATS,Keong93Glide,Liu03Master},
SCATS  attempts to equalize the degree of saturation (DS), 
i.e., the ratio of effectively used green time to the total green time,
for all the approaches.
The computation of cycle length and split plan is only carried out at the critical junctions.
Cycle length and split plan at non-critical junctions are controlled by the critical junctions via offsets.
The algorithm involves many parameters, which need to be properly calibrated for each critical junction.
%For these junctions, SCATS employs a heuristic approach to compute cycle length, with various parameters to be tuned, to achieve this objective.
In addition, all the possible split plans need to be pre-specified and
a voting scheme is used in order to select a split plan that leads to approximately equal DS for all the approaches.

Even though these adaptive traffic signal control systems have been utilized in many cities,
most of them cannot provide any performance guarantee.
Systems and control theory has been recently applied to traffic signal control problems.
In \cite{Diakaki02}, a multivariable regulator is proposed based on linear-quadratic regulator methodology and
the store-and-forward modeling approach \cite{Aboudolas08}.
Robust control theory has been applied to traffic signalization in \cite{Yu97thesis}.
Approaches based on Petri Net modeling language are considered in, e.g., \cite{Mladenovic11thesis,Soares08}.
Optimization-based techniques are considered, e.g., in \cite{Dujardin11,Shen11}.
However, one of the major drawbacks of these approaches is the scalability issue, which limits
their application to relatively small networks.

To address the scalability issue, in \cite{Cheng09}, a distributed algorithm is presented where 
the signal at each junction is locally controlled independently from other junctions.
However, global optimality is no longer guaranteed, although simulation results show that it reduces
the total delay compared to the fixed-time approach.
Another distributed approach is considered in \cite{Lammer08} 
where the constraint that each traffic flow is served once, on average, within a desired service interval $T$ is imposed.
It can be proved that their distributed algorithm
stabilizes the network whenever there exists a stable fixed-time control with cycle time $T$.
However, the knowledge of traffic arrival rates is required.
In addition, multi-phase operation is not considered.

An objective of this work is to develop a traffic signal control strategy that requires minimal tuning
and scales well with the size of the road network while ensuring satisfactory performance.
Our algorithm is motivated by backpressure routing introduced in \cite{Tassiulas92Stability},
which has been mainly applied to communication and power networks 
where a packet may arrive at any node in the network and can only leave the system when it reaches its destination node.
One of the attractive features of backpressure routing is that it leads to maximum network throughput 
without requiring any knowledge about traffic arrival rates
\cite{Tassiulas92Stability,NMR05,Georgiadis06}.

To the authors' knowledge, this is the first time backpressure routing has been adapted
to solve the traffic signal control problem.
Since many assumptions made in backpressure routing are not valid in our traffic signalization application,
certain modifications need to be made to the original algorithm.
With these modifications, we formally prove that our algorithm inherits the desired properties of backpressure routing
as it leads to maximum network throughput even though the signal at each junction is determined completely independently
from the signal at other junctions, and no information about traffic arrival rates is provided.
Furthermore, since our controller is constructed and implemented in a completely distributed manner, 
it can be applied to an arbitrarily large network.
Simulation results show that our algorithm significantly outperforms the SCATS algorithm explained in
\cite{Liu03Master}.

A preliminary version of this work has partially appeared in \cite{Wongpiromsarn:ITSC2012}.
The approach presented in \cite{Wongpiromsarn:ITSC2012}, however, does not allow important constraints such as
periodic switching sequences of phases and minimum and maximum green time to be incorporated.
This paper provides a generalization of \cite{Wongpiromsarn:ITSC2012} by allowing these constraints to be taken into account.
As will be shown later in Section \ref{sec:example}, \cite{Wongpiromsarn:ITSC2012} is a special case of this paper
where the flow rate through a junction is assumed to be constant and no constraints on the minimum and maximum green time
are imposed.

The remainder of the paper is organized as follows: 
We provide useful definitions and existing results concerning network stability in the following section. 
Section \ref{sec:prob} describes the traffic signal control problem considered in this paper.
Our backpressure-based traffic signal control algorithm is described in Section \ref{sec:contr}.
In Section \ref{sec:contr_evaluation}, we formally prove that our algorithm
ensures global optimality as it leads to maximum network throughput, even though the signal at each junction
is determined completely independently from other junctions.
Section \ref{sec:example} presents examples, showing that under a certain assumption on the flow rate,
we can derive the result presented in \cite{Wongpiromsarn:ITSC2012}.
Section \ref{sec:results} provides simulation results, showing that our algorithm offers superior network performance compared
to SCATS.
Finally, Section \ref{sec:discussion} discusses key advantages of the algorithm presented in this paper over existing algorithms 
and Section \ref{sec:conclusions} concludes the paper and discusses future work.

%%%%%%%%%%%%%%%%%%%%%%%%%%%%%%%%%%%%%%%%%%%%%%%%%%%%%%%%%%%%%%%
\section{Preliminaries}
Let $\naturals = \{0, 1, \ldots\}$ be the set of natural numbers, including 0.
In this section, we summarize existing results and definitions concerning network stabilility.
We refer the reader to \cite{Tassiulas92Stability,NMR05,Georgiadis06} for more details.

Consider a network modeled by a directed graph with $N$ nodes and $L$ links.
Each node maintains an internal queue of objects to be processed by the network, while
each link $(a,b)$ represents a channel for direct transmission of objects from node $a$ to node $b$.
Suppose the network operates in slotted time $t \in \naturals$. %where $\naturals$ is the set of natural numbers (including zero).
Objects may arrive at any node in the network and can only leave the system upon reaching the their destination node.
Let $A_i(t)$ represent the number of objects that exogenously arrives at source node $i$ during slot $t$ and
$U_i(t)$ represent the queue length at node $i$ at time $t$.
We assume that all the queues have infinite capacity.
In addition, only the objects currently at each node at the beginning of slot $t$ can be transmitted during that slot.
Our control objective is to ensure that all queues are stable as defined below. %so that the queue length is uniformly bounded.

\begin{definition}
A network is \emph{strongly stable} if each individual queue $U$ satisfies
\begin{equation}
\limsup_{t \to \infty} \frac{1}{t} \sum_{\tau=0}^{t-1} 1_{[U(\tau) > V]} \to 0 \hbox{ as } V \to \infty,
\end{equation}
where for any event $X$, the indicator function $1_X$ takes the value 1 if X is satisfied 
and takes the value 0 otherwise.
\end{definition}

%\begin{definition}
%A network is \emph{strongly stable} if each individual queue $U$ of the network satisfies
%\begin{equation*}
%\limsup_{t \to \infty} \frac{1}{t} \sum_{\tau=0}^{t-1} \expect\{U(\tau)\} < \infty.
%\end{equation*}
%\end{definition}

In this paper, we restrict our attention to strong stability and use the term ``stability'' to refer to strong stability defined above.
%For a network with $K$ queues $U_1, \ldots, U_K$, we define
%\begin{eqnarray}
%g_k(V) &=& \limsum_{t \to \infty} \frac{1}{t} \sum_{\tau=0}^{t-1} 1_{[U_k(\tau) > V]} \hbox{ for each } k \in \{1, \ldots, K\},\\
%g_{sum}(V) &=& \limsum_{t \to \infty} \frac{1}{t} \sum_{\tau=0}^{t-1} 1_{[U_1(\tau) + \ldots + U_K(\tau) > V]}.
%\end{eqnarray}
%
%The following necessary condition for network stability has been proved in, e.g., \cite{}.
%\begin{proposition}
%If a network is stable, then the probability that the unfinished work in all queues 
%simultaneously drops below $V$ is greater than $1/2$ infinitely often.
%\end{proposition}
%
For a network with $N$ queues $U_1, \ldots, U_N$ that evolve according to some probabilistic law, 
a sufficient condition for stability can be provided using Lyapunov drift.
%the following Lyapunov-based argument provides a sufficient condition for stability

\begin{proposition}
\label{prop:LyapunovStability}
Suppose $\expect\{U_i(0)\} < \infty$ for all $i \in \{1, \ldots, N\}$
and there exist constants $B > 0$ and $\epsilon > 0$ such that
\begin{equation}
\expect \Big\{L(\mathbf{U}(t+1)) - L(\mathbf{U}(t)) \Big| \mathbf{U}(t)\Big\} \leq B - \epsilon \sum_{i=1}^N U_i(t), \forall t \in \naturals,
\end{equation}
where for any queue vector $\mathbf{U} = [U_1, \ldots, U_N]$, $L(\mathbf{U}) \triangleq \sum_{i=1}^N U_i^2$.
Then the network is strongly stable.
\end{proposition}

\begin{definition}
\label{def:rate}
An arrival process $A(t)$ is \emph{admissible with rate} $\lambda$ if:
\begin{itemize}
\item The time average expected arrival rate satisfies
\begin{equation*}
\lim_{t \to \infty} \frac{1}{t} \sum_{\tau=0}^{t-1} \expect\{A(\tau)\} = \lambda.
\end{equation*}
\item There exists a finite value $A_{max}$ such that $\expect\{A(t)^2 \hspace{1mm}|\hspace{1mm} \mathbf{H}(t)\} \leq A_{max}^2$
for any time slot $t$, where $\mathbf{H}(t)$ represents the history up to time $t$, i.e., all events that take place during slots $\tau \in \{0,\ldots,t-1\}$.
\item For any $\delta > 0$, there exists an interval size $T$ (which may depend on $\delta$) such that for any initial time $t_0$,
\begin{equation*}
\expect\left\{ \frac{1}{T}  \sum_{k=0}^{T-1} A(t_0 + k) \hspace{1mm}\Big|\hspace{1mm} \mathbf{H}(t_0) \right\} \leq \lambda + \delta.
\end{equation*}
\end{itemize}
\end{definition}

For each node $i$, we define $\lambda_i$ to be the time average rate with which 
$A_i(t)$ is admissible.
%Assume that for each commodity $c$ and node $i$, $A_i^{(c)}(t)$ is admissible with time average rate $\lambda_i^{(c)}$.
Let $\boldsymbol{\lambda} = \left[\lambda_i\right]$ represent the arrival rate vector.

%\begin{definition}
%The \emph{network layer capacity region} $\Lambda$ is the closure of the set of all arrival rate matrices $\boldsymbol{\lambda}$	
%that can be stably supported by the network, considering all possible strategies for choosing the control variables to affect routing, scheduling, and resource allocation (including strategies that have perfect knowledge of future events).
%\end{definition}

\begin{definition}
\label{def:capacity region}
The \emph{capacity region} $\Lambda$ is the closed region of arrival rate vectors $\boldsymbol{\lambda}$
with the following properties:
\begin{itemize}
\item $\boldsymbol{\lambda} \in \Lambda$ is a necessary condition for network stability, considering
all possible strategies for choosing the control variables %to affect routing, scheduling, and resource allocation 
(including strategies that have perfect knowledge of future events).
\item $\boldsymbol{\lambda}$ strictly interior to $\Lambda$ %strictly interior to $\Lambda$ 
is a sufficient condition for the network to be stabilized by a
policy that does not have a-priori knowledge of future events.
\end{itemize}
\end{definition}

The capacity region essentially describes the set of all arrival rate vectors that can be stably supported by the network.
A scheduling algorithm is said to maximize the network throughput if it stabilizes 
the network for all arrival rates in the interior of $\Lambda$.

%%%%%%%%%%%%%%%%%%%%%%%%%%%%%%%%%%%%%%%%%%%%%%%%%%%%%%%%%%%%%%%
\section{The Traffic Signal Control Problem}
\label{sec:prob}
We consider a road network with $N$ links and $L$ signalized junctions.
Specifically, we define a road network as a tuple
$\N = (\R, \J)$ where $\R = \{\R_1, \ldots, \R_{N}\}$
and $\J = \{ \J_1, \ldots, \J_{L}\}$ are sets of all the links and signalized junctions, respectively, in $\N$.
A traffic movement through junction $\J_i, i \in \{1, \ldots, L\}$ is defined as a pair $(\R_a, \R_b)$ where $\R_a, \R_b \in \R$ 
such that a vehicle may enter and exit $\J_i$ through $\R_a$ and $\R_b$, respectively.
A phase of $\J_i$ is defined as a set of traffic movements that simultaneously receiving the right-of-way.

Junction $\J_i, i \in \{1, \ldots, L\}$ is defined by a tuple $\J_i = (\M_i, \P_i, \Z_i)$ where 
$\M_i \subseteq \R^2$ is a set of all the possible traffic movements through $\J_i$, 
$\P_i \subseteq 2^{\M_i}$ is a set of all the possible phases of $\J_i$ and
$\Z_i$ is a finite set of traffic states, each of which captures factors
that affect the traffic flow rate through $\J_i$.
These factors may include, but not limited to, the number of vehicles on the relevant links, road disruptions, traffic and weather conditions.
A typical set of phases of a 4-way junction is shown in Figure \ref{fig:phase-ex}.

\begin{figure}
   \centering 
   \subfigure[]
        {\includegraphics[trim=5.5cm 10cm 13cm 4cm, clip=true, width=0.2\textwidth]{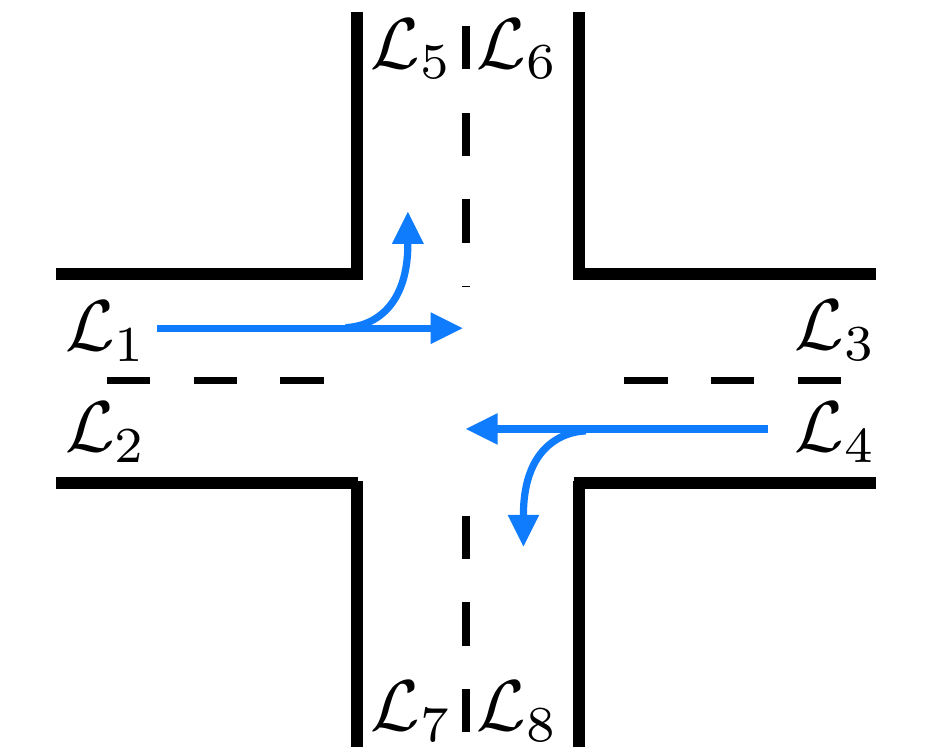}}
   \subfigure[]
        {\includegraphics[trim=5.5cm 10cm 13cm 4cm, clip=true, width=0.2\textwidth]{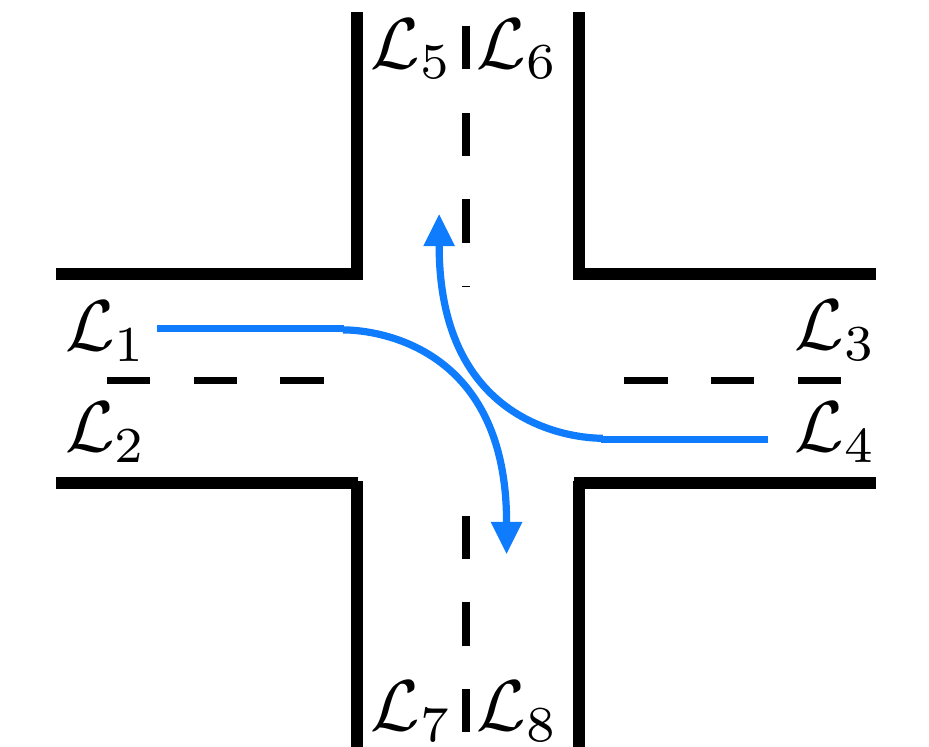}}
   \subfigure[]
        {\includegraphics[trim=5.5cm 10cm 13cm 4cm, clip=true, width=0.2\textwidth]{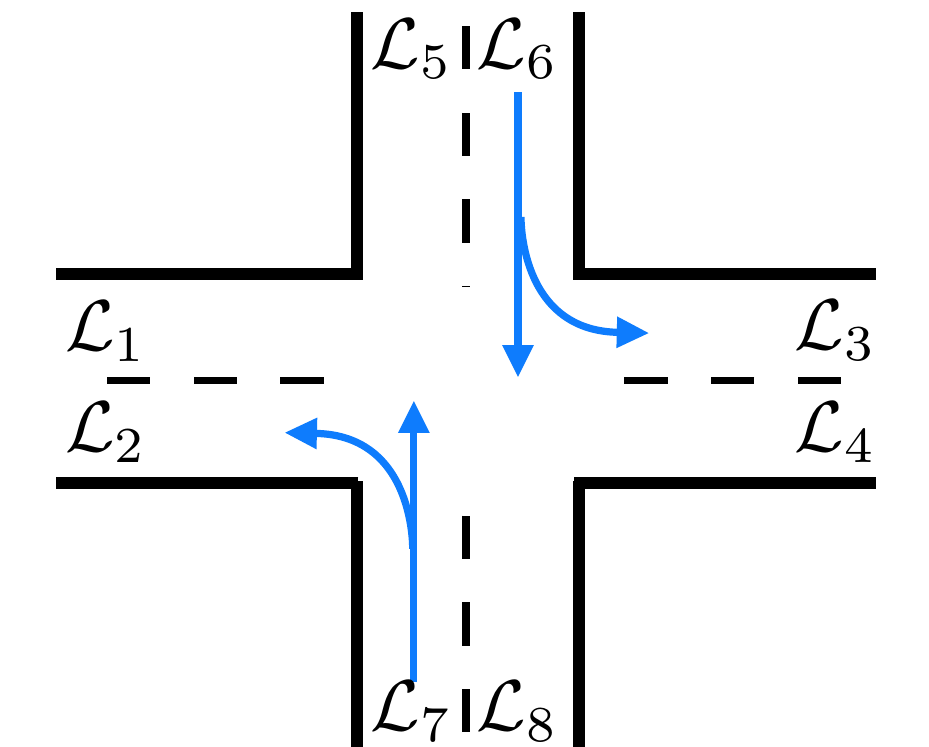}}
   \subfigure[]
        {\includegraphics[trim=5.5cm 10cm 13cm 4cm, clip=true, width=0.2\textwidth]{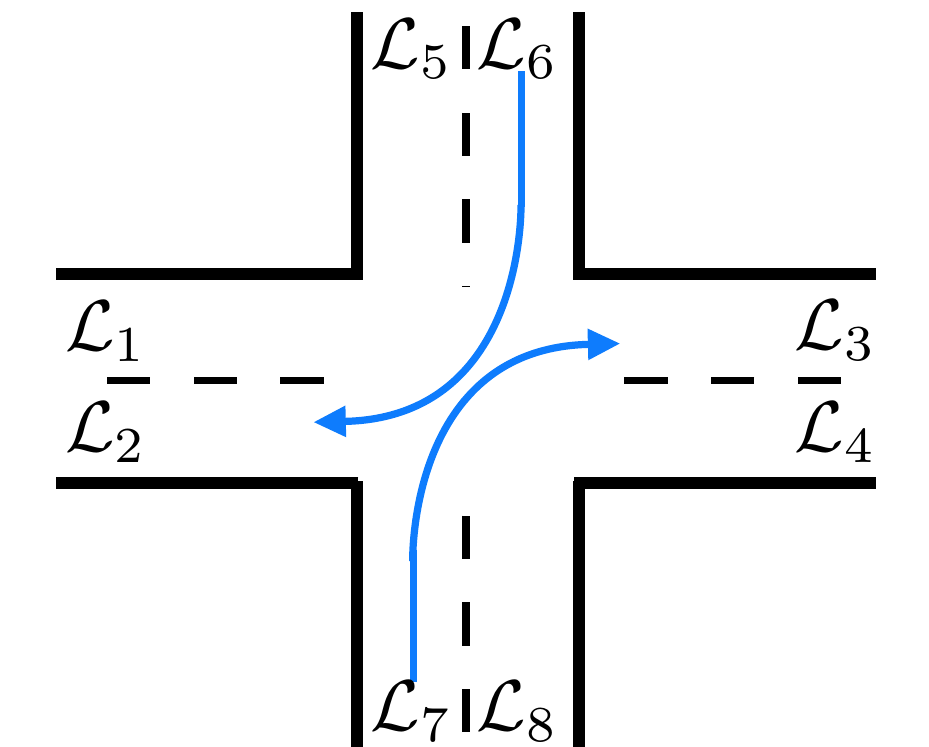}}
   \caption{A typical set $\{\P_1, \P_2, \P_3, \P_4\}$ of phases of a 4-way junction with links $\R_1, \ldots, \R_8$.
   (a) $\P_1 = \{(\R_1, \R_3), (\R_1, \R_5), (\R_4, \R_2), (\R_4, \R_8)\}$,
   (b) $\P_2 = \{(\R_1, \R_8), (\R_4, \R_5)\}$,
   (c) $\P_3 = \{(\R_7, \R_5), (\R_7, \R_2)$, $(\R_6, \R_8), (\R_6, \R_3)\}$, and
   (d) $\P_4 = \{(\R_7, \R_3), (\R_6, \R_2)\}$.}
  \label{fig:phase-ex}
\end{figure}

Vehicles may enter the network at any link at any time.
The traffic signal controller operates in slotted time $t \in \naturals$, monitoring the traffic
and dynamically sets the timing of traffic signals.
Specifically, at the beginning of each time slot, 
the traffic signal controller determines a split plan for each junction. 
A split plan for junction $\J_i$ is defined as a function
$S_i : \P_i \to [0,1]$ that maps each phase $p$ of $\J_i$ 
to the amount of time slot to be allocated to phase $p$ such that
$\sum_{p \in \P_i} S_i(p) = 1$ and for all $p \in \P_i$, $\underline{T}^i_p \leq S_i(p) \leq \overline{T}^i_p$
where $\underline{T}^i_p \in [0,1]$ and $\overline{T}^i_p \in [0,1]$ are the minimum and maximum amount of
time slot that can be allocated to phase $p$.
We let $\splitplans_i = \{S_i : \P_i \to [0,1] \hspace{1mm}|\hspace{1mm} 
\sum_{p \in \P_i} S_i(p) = 1 \hbox{ and } \underline{T}^i_p \leq S_i(p) \leq \overline{T}^i_p, \forall p \in \P_i\}$
be the set of all the possible split plans for junction $\J_i$.
We assume that for each $i \in \{1, \ldots, N\}$, 
$\splitplans_i \not= \emptyset$, i.e.,
there exists a valid split plan for each junction.
Note that this assumption can be satisfied by ensuring that
$\sum_{p \in \P_i} \underline{T}^i_p \leq 1$ and
$\sum_{p \in \P_i} \overline{T}^i_p \geq 1$.

For each $a \in \{1, \ldots, N\}$, $i \in \{1, \ldots, L\}$, $t \in \naturals$, 
we let $Q_a(t) \in \naturals$ and $z_i(t) \in \Z_i$ represent the number of vehicles on $\R_a$
and the traffic state at junction $\J_i$, respectively, at the beginning of time slot $t$.
In addition, for each $a,b \in \{1, \ldots, N\}$, $i \in \{1, \ldots, L\}$ and $p \in \P_i$,
we define a function $\xi^{i}_{a,b} : \splitplans_i \times \Z_i \to \naturals$ such that
$\xi^{i}_{a,b}(S, z)$ gives the number of vehicles 
that can go from $\R_a$ to $\R_b$ through junction $\J_i$ in one time slot
under traffic state $z$ and split plan $S$. %if phase $p$ is activated for duration $d$ of the time slot.
%By definition, $\xi_{i,p}^{a,b}(d, z) = 0$ for all $z \in \Z_i$ and $d \in [0,1]$ if $(\R_a, \R_b) \not\in p$, i.e., 
%phase $p$ does not give the right of way to the traffic movement from $\R_a$ to $\R_b$.
%
When traffic state $z$ represents the case where the number of vehicles on $\R_a$ that 
seek the movement to $\R_b$ through $\J_i$ is large, 
$\xi^{i}_{a,b}(S, z)$ can be simply obtained by assuming saturated flow.
Note that $\xi^{i}_{a,b}$ does not need to be known a priori
but we assume that it is available at the beginning of each time slot.
%In addition, $\xi^{i}_{a,b}$ may be different for different time slots, e.g.,
%it may depend on the split plan of the previous time slot.
%it may depend on the phase activated in the previous time slot.

In this paper, we consider the traffic signal control problem as stated below.

\textbf{Traffic Signal Control Problem:} Design a traffic signal controller that determines the split plan $S_i \in \splitplans_i$
for each junction $\J_i, i \in \{1, \ldots, L\}$ during each time slot $t \in \naturals$ such that the network throughput is maximized.
We assume that there exists a reliable traffic monitoring system that provides the queue length $Q_a(t)$
and traffic state $z_i(t)$ for each $a \in \{1, \ldots, N\}$, $i \in \{1, \ldots, L\}$  %$(\R_a, \R_b) \in \M_i$ 
at the beginning of each time slot $t \in \naturals$ to the controller.

%%%%%%%%%%%%%%%%%%%%%%%%%%%%%%%%%%%%%%%%%%%%%%%%%%%%%%%%%%%%%%%
\section{Backpressure-based Traffic Signal Controller}
\label{sec:contr}

In this section, we propose a distributed traffic signal control algorithm %that essentially implements 
that employs the idea from backpressure routing %for a single-commodity network 
as described in \cite{Tassiulas92Stability,NMR05,Georgiadis06}.
Unlike most of the traffic signal controllers considered in existing literature, 
our controller can be constructed and implemented in a completely distributed manner,
i.e., the split plan at each junction is determined independently from other junctions,
using only local information, namely the queue length on each of the links associated 
with this junction and the current traffic state around this junction.
No explicit coordination with other junctions is required.
Furthermore, it does not require any knowledge about traffic arrival rates.

Roughly, for each junction $\J_i$, our algorithm picks a split plan that maximizes the ``pressure relief'' at $\J_i$.
The ``pressure relief'' at junction $\J_i$ is defined as 
the sum of the ``pressure relief'' associated with each traffic movement through $\J_i$.
Here, the ``pressure relief'' associated with traffic movement $(\R_a, \R_b)$ is defined as 
the number of vehicles that can go from $\R_a$ to $\R_b$ in one time slot, weighted by
the difference between the number of vehicles on $\R_a$ and on $\R_b$.  
Specifically, consider an arbitrary time slot $t \in \naturals$.
The ``pressure relief'' at $\J_i$ under a split plan $S \in \splitplans_i$
is defined by
\begin{equation}
  \label{eq:PR}
  \pr_i^t(S) \triangleq \sum_{\scriptsize \begin{array}{c}a,b \hbox{ s.t.}\\ (\R_a, \R_b) \in \M_i\end{array}} W_{a,b}(t) \xi^{i}_{a,b}(S, z_i(t)),
\end{equation}
where the weight $W_{a,b}(t)$ is defined for each pair $(\R_a, \R_b) \in \M_i$ by
\begin{equation}
  \label{eq:Wab_traffic}
  W_{a,b}(t) \triangleq Q_a(t) -  Q_b(t). %\max\{ Q_a(t) -  Q_b(t), 0\},
\end{equation}

Based on the above description of the backpressure-based traffic signal control algorithm,
our traffic signal controller thus consists of a set of local controllers $\controller_1, \ldots, \controller_L$
where the local controller $\controller_i$ is associated with junction $\J_i$.
These local controllers are constructed and implemented independently
of one another.
(However, a synchronized operation among all the junctions is required so that 
control actions for all the junctions take place according to a common time clock.)
Furthermore, each local controller does not require the global view of the road network.
Instead, it only requires information that is local to the junction with which it is associated.
Consider an arbitrary junction $\J_i \in \J$ and time slot $t \in \naturals$.
The local controller $\controller_i$ picks a split plan $S^*$ for junction $\J_i$ as a solution of the following optimization problem
\begin{equation}
\label{eq:split_comp}
  \max_{S \in \splitplans_i} \pr_i^t(S).
\end{equation}
That is, the controller picks a split plan $S^*$ such that $\pr_i^t(S^*) \geq \pr_i^t(S)$ for all split plans $S \in \splitplans_i$. 
If there exist multiple options of such $S^*$, the controller can pick one arbitrarily.
Note that from Tychonoff theorem, for all $i \in \{1, \ldots, L\}$,
the set $\splitplans_i$ is a compact topological space with the metric induced by the uniform distance
(i.e., the distance between split plans $S$ and $\tilde{S}$ is defined by
$\max\{|S(p) - \tilde{S}(p)| : p \in \P_i\}$).
Hence, if $\xi_{a,b}^i$ is continuous with respect to its first argument for all $a,b \in \{1, \ldots, N\}$ and $i \in \{1, \ldots, L\}$,
then according to the extreme value theorem, $\pr_i^t(S)$ attains its maximum.

Our algorithm is similar in nature to backpressure routing for a single-commodity network.
In \cite{Tassiulas92Stability,NMR05,Georgiadis06}, it has been shown that backpressure routing
leads to maximum network throughput.
However, it is still premature to simply conclude that our backpressure-based traffic signal control algorithm 
inherits this property due to the following reasons.
First, backpressure routing requires that a commodity at least defines the destination of the object.
Implementing the algorithm for a single-commodity network implies that we assume that all the vehicles have a common destination,
which is not a valid assumption for our application.
Second, backpressure routing assumes that the controller has complete control over routing of the traffic around the network whereas
in our traffic signal control problem, the controller does not have control over the route picked by each driver.
Third, backpressure routing assumes that the network controller has control over the flow rate
of each link subject to the maximum rate imposed by the link constraint.
However, the traffic signal controller can only picks a split plan $S_i$ for each junction $\J_i$ but 
does not have control over the flow rate of each traffic movement once $S_i$ is activated.
To account for this lack of control authority, we slightly modify the definition of $W_{a,b}(t)$ from
that used in backpressure routing.
%Notice the difference in the definition of $W_{ab}(t)$ in (\ref{eq:Wab_backpressure}) and in (\ref{eq:Wab_traffic})
%to account for this lack of control authority.
Finally, the optimality result of backpressure routing relies on the assumption that
all the queues have infinite buffer storage space.
Even though it is not reasonable to assume that all the links have infinite queue capacity,
for the rest of the paper, we assume that this is the case.
In practice, our algorithm is expected to work well when each link can accommodate a reasonably long queue.

Before evaluating the performance of our algorithm, 
we first provide its basic property, which is similar to the basic property of backpressure routing. % as formally stated in the following lemma.
%The detailed proof can be found in a technical report \cite{Wongpiromsarn:DTS2012}.
%

Let $\Z = \Z_1 \times \ldots \times \Z_L$.
For each $a \in \{1, \ldots, N\}$, we define functionals
$V^{out}_a : \splitplans_1 \times \ldots \times \splitplans_L \times \Z \to \reals$ and 
$V^{in}_a : \splitplans_1 \times \ldots \times \splitplans_L  \times \Z \to \reals$ 
such that for any split plan $S_1, \ldots, S_L$ and traffic state $\mathbf{z} \in \Z$,
\begin{equation}
\begin{array}{rcl}
V^{out}_a(S_1, \ldots, S_L, \mathbf{z}) &=&  
\displaystyle{\sum_{\scriptsize \begin{array}{c}b,i \hbox{ s.t.}\\ (\R_a, \R_b) \in \M_i \end{array}} \hspace{-6mm} \xi^{i}_{a,b}(S_i, z_i)},\\
V^{in}_a(S_1, \ldots, S_L, \mathbf{z}) &=& 
\displaystyle{\sum_{\scriptsize \begin{array}{c}b,i \hbox{ s.t.}\\ (\R_b, \R_a) \in \M_i \end{array}} \hspace{-6mm} \xi^{i}_{b,a}(S_i,  z_i)},
\end{array}
\end{equation}
where for each $i \in \{1, \ldots, L\}$,
$z_i \in \Z_i$ is the element of $\mathbf{z}$ that corresponds to the traffic state of junction $\J_i$.

\begin{lemma}
\label{lem:basic_prop}
Consider an arbitrary time slot $t \in \naturals$.
Let $\mathbf{z} \in \Z$ be a vector of traffic states of
all the junctions during time slot $t$.
For each $i \in \{1, \ldots, L\}$, 
let $S_i^*$ denote a split plan for junction $\J_i$ that is a solution of (\ref{eq:split_comp})
and $\tilde{S}_i$ be an arbitrary split plan for junction $\J_i$.
Then, 
\begin{equation}
\label{eq:lem:basic_prop}
\begin{array}{l}
\displaystyle{\sum_a Q_a(t) \Big( V^{out}_a \big(\tilde{S}_1, \ldots, \tilde{S}_L, \mathbf{z} \big) - 
V^{in}_a \big(\tilde{S}_1, \ldots, \tilde{S}_L, \mathbf{z}\big) \Big)} \\
\hspace{8mm}\leq
\displaystyle{\sum_a Q_a(t) \Big( V^{out}_a \big(S_1^*, \ldots, S_L^*, \mathbf{z} \big) - 
V^{in}_a \big(S_1^*, \ldots, S_L^*, \mathbf{z} \big) \Big)}.
\end{array}
\end{equation}
\end{lemma}

\begin{proof}
%First, we note the following identity
%\begin{equation}
%\label{pf:basic_prop1}
%\begin{array}{rcl}
%&&\sum_a Q_a(t) \left( 
%\displaystyle{\sum_{\scriptsize \begin{array}{c}b,i \hbox{ s.t. }\\ (\R_a, \R_b) \in \M_i \end{array}} \hspace{-6mm} \xi^{i}_{a,b}(S_i, z_i) 
%\hspace{3mm} -
%\hspace{-6mm}
%\sum_{\scriptsize \begin{array}{c}c,i \hbox{ s.t. }\\ (\R_c, \R_a) \in \M_i \end{array}} \hspace{-6mm} \xi^{i}_{c,a}(S_i, z_i)} 
%\right)\\
%&=&
%\displaystyle{\sum_{\scriptsize \begin{array}{c}a,b,i \hbox{ s.t. }\\ (\R_a, \R_b) \in \M_i \end{array}} \xi^{i}_{a,b}(S_i, z_i)Q_a(t) 
%\hspace{3mm} -
%\sum_{\scriptsize \begin{array}{c}c,a,i \hbox{ s.t. }\\ (\R_c, \R_a) \in \M_i \end{array}} \xi^{i}_{c,a}(S_i, z_i)Q_a(t)}\\
%&=&
%\displaystyle{\scriptsize \sum_{\begin{array}{c}c,b,i \hbox{ s.t. }\\ (\R_c, \R_b) \in \M_i \end{array}} \xi^{i}_{c,b}(S_i, z_i)Q_c(t)
%\hspace{3mm} -
%\sum_{\scriptsize \begin{array}{c}c,b,i \hbox{ s.t. }\\ (\R_c, \R_b) \in \M_i \end{array}} \xi^{i}_{c,b}(S_i, z_i)Q_b(t)}\\
%&=&
%\displaystyle{\sum_{\scriptsize \begin{array}{c}c,b,i \hbox{ s.t. }\\ (\R_c, \R_b) \in \M_i \end{array}} \hspace{-6mm} \xi^{i}_{c,b}(S_i, z_i)
%\left(Q_c(t) - Q_b(t)\right)},
%\end{array}
%\end{equation}
%
First, we note the following identity
\begin{equation}
\label{pf:basic_prop1}
\begin{array}{l}
\displaystyle{\sum_a Q_a(t) \Big( V^{out}_a\big(S_1, \ldots, S_L, \mathbf{z}\big) - V^{in}_a\big(S_1, \ldots, S_L, \mathbf{z}\big) \Big)}\\
\hspace{8mm}=\hspace{-5mm}
\displaystyle{\sum_{\scriptsize \begin{array}{c}a,b,i \hbox{ s.t. }\\ (\R_a, \R_b) \in \M_i \end{array}} \hspace{-6mm} \xi^{i}_{a,b} \big(S_i, z_i \big)
W_{a,b}(t)},
\end{array}
\end{equation}
for any time slot $t \in \naturals$, split plan $S_1, \ldots, S_L$ and traffic state $\mathbf{z} \in \Z$.

Since for each $i \in \{1, \ldots, L\}$, $S_i^*$ is chosen such that $\pr_i^t(S_i^*) \geq \pr_i^t(\tilde{S}_i)$,
we get
\begin{equation}
\label{pf:basic_prop2}
\begin{array}{l}
\displaystyle{\sum_{\scriptsize \begin{array}{c}a,b \hbox{ s.t. }\\ (\R_a, \R_b) \in \M_i \end{array}} \hspace{-6mm} \xi^{i}_{a,b} \big(\tilde{S}_i, z_i \big)
W_{ab}(t)} 
\hspace{3mm}\leq
%\hspace{8mm}\leq
\displaystyle{\sum_{\scriptsize \begin{array}{c}a,b \hbox{ s.t. }\\ (\R_a, \R_b) \in \M_i \end{array}} \hspace{-6mm} \xi^{i}_{a,b} \big(S_i^*, z_i \big)
W_{ab}(t)},
\end{array}
\end{equation}
for all $i \in \{1, \ldots, L\}$.
The result in (\ref{eq:lem:basic_prop}) can be obtained by summing the inequality in (\ref{pf:basic_prop2}) over
$i \in \{1, \ldots, L\}$ and using the identity in (\ref{pf:basic_prop1}).
\end{proof}

%%%%%%%%%%%%%%%%%%%%%%%%%%%%%%%%%%%%%%%%%%%%%%%%%%%%%%%%%%%%%%%
\section{Controller Performance Evaluation}
\label{sec:contr_evaluation}

Let $\Lambda$ be the capacity region of the road network as defined in Definition \ref{def:capacity region}.
Assume that $\mathbf{z}(t) = \left[ z_i(t) \right]$ evolve according to a finite state, irreducible, aperiodic Markov chain.
Let $\pi_\mathbf{z}$ represent the time average fraction of time that $\mathbf{z}(t) = \mathbf{z}$, i.e., with probability 1, we have
$\lim_{t \to \infty} \frac{1}{t} \sum_{\tau = 0}^{t-1} 1_{[\mathbf{z}(\tau) = \mathbf{z}]} = \pi_\mathbf{z}$, for all 
$\mathbf{z} \in \Z$
where $1_{[\mathbf{z}(\tau) = \mathbf{z}]}$ is an indicator function that takes the value 1 if $\mathbf{z}(\tau) = \mathbf{z} $ 
and takes the value 0 otherwise.
In addition, we let $\M = \bigcup_i \M_i$ be the set of all the possible traffic movements. %, i.e., 
%for each $m \in \M$, there exists $i \in \{1, \ldots, L\}$ such that $m \in \M_i$.
%For the simplicity of the presentation, we assume that such $i$ is unique for each $m \in \M$.
For the simplicity of the presentation, we assume that $\M_i \cap \M_j = \emptyset$ for all $i \not= j$.
For each $S_i \in \splitplans_i, i \in \{1, \ldots, L\}$ and $\mathbf{z} \in \Z$,
we define a vector $\boldsymbol{\xi}(S_1, \ldots, S_L, \mathbf{z})$ whose $k^{th}$ element is equal to
$\xi^{i}_{a,b}(S_i, z_i)$ where $(\R_a, \R_b)$ is the $k^{th}$ traffic movement in $\M$,
$i$ is the (unique) index satisfying $(\R_a, \R_b) \in \M_i$ and 
$z_i$ is the $i^{th}$ element of $\mathbf{z}$.
Define
\begin{equation}
\label{eq:Gamma}
\Gamma \triangleq \sum_{\mathbf{z} \in \Z} \pi_\mathbf{z} 
\conv \Big\{ \left[ \boldsymbol{\xi}(S_1, \ldots, S_L, \mathbf{z}) \right] 
\hspace{1mm}\Big|\hspace{1mm} S_i \in \splitplans_i \hbox{ for all } i \in \{1, \ldots, L\} \Big\},
\end{equation}
where for any set $\mathcal{S}$, $\conv\{\mathcal{S}\}$ represents the convex hull of $\mathcal{S}$.

Additionally, we assume that the process of vehicles exogenously entering the network is rate ergodic.
%and for all for all $a \in \{1, \ldots, N\}$, there are always enough vehicles on $\R_a$ such that 
%for all $i \in \{1, \ldots, L\}$, $b \in \{1, \ldots, N\}$, $p \in \P_i$, $z \in \Z_i$ such that $(\R_a, \R_b) \in \M_i$, 
%vehicles can move from $\R_a$ to $\R_b$ through junction $\J_i$ 
%at rate  $\xi_i(p, \R_a, \R_b, z)$ under traffic state $z$ if phase $p$ is activated at $\J_i$.
For each $a \in \{1, \ldots, N\}$, let $\lambda_a$ be the time average rate with which the number of new vehicles that exogenously enter
the network at link $\R_a$  during each time slot is admissible.
Let $\boldsymbol{\lambda} = \left[ \lambda_a \right]$ represent the arrival rate vector.

Before deriving the optimality result for our backpressure-based traffic signal control algorithm,
we first characterize the capacity region of the road network, as formally stated in the following lemma.
%The proof can be found in \cite{Wongpiromsarn:DTS2012}.

\begin{lemma}
The capacity region of the network is given by the set $\Lambda$ consisting of
all the rate vectors $\boldsymbol{\lambda}$ such that
there exists a rate vector $\mathbf{G} \in \Gamma$ together with flow variables
$f_{ab}$ for all $a,b \in \{1, \ldots, N\}$ satisfying
%Suppose the network can be stabilized by some policy
%(possibly one that bases its decisions upon complete knowledge of future arrivals) 
%when the arrival rate vector is $\boldsymbol{\lambda}$.
%Then, there exists a rate vector $\mathbf{G} \in \cl\{\Gamma\}$ together with flow variables
%$f_{ab}$ for all $a,b \in \{1, \ldots, N\}$ satisfying
\begin{eqnarray}
\label{eq:f-nonnegative}
f_{ab} \geq 0, &&\forall a,b \in \{1, \ldots, N\},\\
\label{eq:f-conservation}
\lambda_a = \sum_b f_{ab} - \sum_c f_{ca}, &&\forall a \in \{1, \ldots, N\},\\
\label{eq:f-zero}
f_{ab} = 0, &&\forall a,b \in \{1, \ldots, N\} \\
\nonumber
&& \hbox{such that } (\R_a, \R_b) \not\in \M,
\end{eqnarray}
\begin{eqnarray}
\label{eq:f-constraint}
f_{ab} = G_{ab}, &&\forall a,b \in \{1, \ldots, N\}  \\
\nonumber
&& \hbox{such that } (\R_a, \R_b) \in \M,
\end{eqnarray}
where $G_{ab}$ is the element of $\mathbf{G}$ that corresponds to the rate of traffic movement
$(\R_a, \R_b)$.
\end{lemma}
\begin{proof}
%Here, we prove that $\boldsymbol{\lambda} \in \Lambda$ is a necessary condition for network stability, 
%considering all possible strategies for choosing the control variables (including strategies that have perfect knowledge of future events).
%Later in this section, we show that $\boldsymbol{\lambda}$ strictly interior to $\Lambda$ is a sufficient
%condition for the network to be stabilized by proving that there exists a stationary randomized control algorithm that
%stabilizes the network under such arrival rates.
%
First, we prove that $\boldsymbol{\lambda} \in \Lambda$ is a necessary condition for network stability, 
considering all possible strategies for choosing the control variables (including strategies that have perfect knowledge of future events).
Consider an arbitrary time slot $t$.
For each $a \in \{1, \ldots, N\}$, let $X_a(t)$ denote the total number of vehicles that exogenously enters the road network
at link $\R_a$ during time slots $0, \ldots, t-1$.
Suppose the network can be stabilized by some policy,
possibly one that bases its decisions upon complete knowledge of future arrivals.
For each $a,b \in \{1, \ldots, N\}$, let $Q_a(t)$ represent the number of vehicles left on $\R_a$ at the beginning of time slot $t$ 
and $F_{ab}(t)$ represent the total number of vehicles executing the $(\R_a, \R_b)$ movement during time slots $0, \ldots, t-1$ 
under this stabilizing policy.
Due to flow conservation and link constraints, we have
\begin{equation}
\label{eq:F-nonnegative}
F_{ab}(t) \geq 0,
\end{equation}
\begin{equation}
\label{eq:F-conservation}
X_a(t) - Q_a(t) = \displaystyle{\sum_b F_{ab}(t) - \sum_c F_{ca}(t)},
\end{equation}
\begin{equation}
\label{eq:F-constraint}
F_{ab}(t) = \left\{
\begin{array}{ll}
0, &\hspace{-2mm}\hbox{if } (\R_a, \R_b) \not\in \M,\\
\hspace{-2mm}\displaystyle{\sum_{\tau=0}^{t-1} \xi^{i}_{a,b}(S_i(\tau), z_i(\tau))}, &\hspace{-2mm}\hbox{if } (\R_a, \R_b) \in \M_i
\end{array} \right.
\end{equation}
for all $a,b \in \{1, \ldots, N\}$
where $S_i(\tau)$ and $z_i(\tau)$ are the split plan and traffic state, respectively, of junction $\J_i$ at time slot $\tau$.

For each $a,b \in \{1, \ldots, N\}$, define $f_{ab} \triangleq F_{ab}(\tilde{t})/\tilde{t}$ for some arbitrarily large time $\tilde{t}$.
It is clear from (\ref{eq:F-nonnegative}) and (\ref{eq:F-constraint}) that (\ref{eq:f-nonnegative}) and (\ref{eq:f-zero}) are satisfied.
In addition, we can follow the proof in \cite{NMR05} to show that there exists a sample paths $F_{ab}(t)$
such that $f_{ab}$ comes arbitrarily close to satisfying (\ref{eq:f-conservation}) and (\ref{eq:f-constraint}).
As a result, it can be shown that $\boldsymbol{\lambda}$ is a limit point of the capacity region $\Lambda$.
%Hence, it follows that they can be satisfied if each nonzero entry of the rate vector $\boldsymbol{\lambda}$
%is reduced by an arbitrarily small amount.
%This proves that $\boldsymbol{\lambda}$ is a limit point of the capacity region $\Lambda$.
Since $\Lambda $ is compact and hence contains its limit points, it follows that $\boldsymbol{\lambda} \in \Lambda$.

Next, we show that $\boldsymbol{\lambda}$ strictly interior to $\Lambda$ is a sufficient condition for network stability,
considering only strategies that do not have a-priori knowledge of future events.
Suppose the rate vector $\boldsymbol{\lambda}$ is such that there exists $\boldsymbol{\epsilon} > 0$ such that
$\boldsymbol{\lambda} + \boldsymbol{\epsilon} \in \Lambda$.
%Suppose there exist $\mathbf{G} \in \Gamma$ and $f_{ab}$ for all $a,b \in \{1, \ldots, N\}$ satisfying
%(\ref{eq:f-nonnegative})-(\ref{eq:f-constraint}).
Let $\mathbf{G} \in \Gamma$ be a transmission rate vector associated with the input rate vector $\boldsymbol{\lambda} + \boldsymbol{\epsilon}$
according to the definition of $\Lambda$.
It has been proved in \cite{NMR05} that there exists a stationary randomized policy
$\tilde{S}_i(\tau)$ for each $i \in \{1, \ldots, L\}$ that satisfies certain convergence bounds and such that
for each $(\R_a, \R_b) \in \M_i$,
$\lim_{t \to \infty} \frac{1}{t} \sum_{\tau=0}^{t-1} \xi^{i}_{a,b}(\tilde{S}_i(\tau), z_i(\tau)) = G_{ab}$.
In addition, such a policy stabilizes the system.
\end{proof}

\begin{corollary}
\label{cor:randomized}
Suppose %$\Gamma = \cl\{\Gamma\}$ and if 
the traffic state $\mathbf{z}$ is i.i.d. from slot to slot.
Then, $\boldsymbol{\lambda}$ is within the capacity region $\Lambda$
if and only if there exists a stationary randomized control algorithm that determines split plans $S_1, \ldots, S_L$
based only on the current traffic state $\mathbf{z}$, and that yields for all $a \in \{1, \ldots, N\}$, $t \in \naturals$,
\begin{equation}
\begin{array}{c}
\expect \Bigg\{ V^{out}_a \big(S_1, \ldots, S_L, \mathbf{z} \big) 
 - V^{in}_a \big(S_1, \ldots, S_L, \mathbf{z} \big) \Bigg\}
= \lambda_a,
\end{array}
\end{equation}
where the expectation is taken with respect to the random traffic state $\mathbf{z}$ and the (potentially) random control action based on this state.
\end{corollary}

Finally, based on the above corollary and the basic property of our backpressure-based traffic signal control algorithm,
we can conclude that our algorithm leads to maximum network throughput.

\begin{theorem}
If there exists $\boldsymbol{\epsilon} > 0$ such that $\boldsymbol{\lambda} + \boldsymbol{\epsilon} \in \Lambda$,
then the proposed backpressure-based traffic signal controller stabilizes the network, provided that
the traffic state $\mathbf{z}$ is i.i.d. from slot to slot.
\end{theorem}
\begin{proof}
Consider an arbitrary policy $\tilde{S}_1, \ldots, \tilde{S}_L$ and time slot $t \in \naturals$.
By simple manipulations, we get
\begin{equation*}
\hspace{-2mm}
\begin{array}{l}
L(\mathbf{Q}(t+1)) - L(\mathbf{Q}(t)) \leq B -  \\ 
\hspace{3mm}2
\displaystyle{\sum_{a} Q_a(t) \Big( V^{out}_a \big(\tilde{S}_1, \ldots, \tilde{S}_L, \mathbf{z}(t) \big) - A_a(t) - V^{in}_a \big( \tilde{S}_1, \ldots, \tilde{S}_L, \mathbf{z}(t) \big) \Big)},
\end{array}
\end{equation*}
where $A_a(t)$ is the number of vehicle that exogenously enter the network at link $\R_a$ during time slot $t$,
\begin{equation*}
\begin{array}{rcl}
B &=&
\displaystyle{\sum_a \Bigg( \Big( \sup_{\scriptsize \begin{array}{c}S_1 \in \splitplans_1, \ldots,\\ S_L \in \splitplans_L, \mathbf{z} \in \Z\end{array}}  
V^{out}_a \big( S_1, \ldots, S_L, \mathbf{z}(t) \big) \Big)^2} +\\
&&\displaystyle{\Big( A_a^{max} + 
\sup_{\scriptsize \begin{array}{c}S_1 \in \splitplans_1, \ldots,\\S_L \in \splitplans_L, \mathbf{z} \in \Z\end{array}}  
V^{in}_a \big( S_1, \ldots, S_L, \mathbf{z}(t) \big) \Big)^2 \Bigg)}
\end{array}
\end{equation*}
and $A_a^{max}$ satisfies $A_a(t) \leq A_a^{max}, \forall t$. Hence, we get
\begin{equation*}
\begin{array}{l}
\expect \Big\{ L(\mathbf{Q}(t+1)) - L(\mathbf{Q}(t)) \Big| \mathbf{Q}(t) \Big\} \leq
\displaystyle{B + 2\sum_a Q_a(t) \expect \Big\{ A_a(t) \Big| \mathbf{Q}(t) \Big\} -}\\
\hspace{12mm} 2 \sum_{a} Q_a(t) \expect\Big\{ V^{out}_a \big( \tilde{S}_1, \ldots, \tilde{S}_L, \mathbf{z}(t) \big) - V^{in}_a \big( \tilde{S}_1, \ldots, \tilde{S}_L, \mathbf{z}(t) \big) \Big| \mathbf{Q}(t)\Big\}
\end{array}
\end{equation*}

However, from Lemma \ref{lem:basic_prop}, the proposed backpressure-based traffic signal controller minimizes
the final term on the right hand side of the above inequality over all possible alternative policies $\tilde{S}_1, \ldots, \tilde{S}_L$.
But since $\boldsymbol{\lambda} + \boldsymbol{\epsilon} \in \Lambda$, according to Corollary \ref{cor:randomized},
there exists a stationary randomized algorithm that makes phase decisions based only on
the current traffic state $\mathbf{z}(t)$ and that yields for all $a \in \{1, \ldots, N\}$, $t \in \naturals$,
\begin{equation*}
\begin{array}{c}
\expect \Big\{ V^{out}_a \big( \tilde{S}_1, \ldots, \tilde{S}_L, \mathbf{z}(t) \big) - V^{in}_a \big( \tilde{S}_1, \ldots, \tilde{S}_L, \mathbf{z}(t) \big) 
 \Big| \mathbf{Q}(t) \Big\}
= \lambda_a + \epsilon.
\end{array}
\end{equation*}
Hence, we get that when the proposed backpressure-based traffic signal controller is used,
\begin{equation*}
\expect \Big\{ L(\mathbf{Q}(t+1)) - L(\mathbf{Q}(t)) \Big | \mathbf{Q}(t) \Big\} \leq B - 2\epsilon\sum_a Q_a(t),
\end{equation*}
and from Proposition \ref{prop:LyapunovStability}, we can conclude that the network is stable.
\end{proof}

%%%%%%%%%%%%%%%%%%%%%%%%%%%%%%%%%%%%%%%%%%%%%%%%%%%%%%%%%%%%%%%
\section{Examples}
\label{sec:example}
In this section, we consider a special case where the flow rate through a junction for each phase is constant and only depends on the traffic state, i.e.,
\begin{equation}
\label{eq:xi_const_flow}
\xi^{i}_{a,b}(S, z) = \sum_{\scriptsize \begin{array}{c}p \in \P_i \hbox{ s.t.}\\ (\R_a, \R_b) \in p\end{array}} \alpha^{i,p}_{a,b}(z) S(p),
\end{equation}
where for each $i \in \{1, \ldots, L\}$, $a,b \in \{1, \ldots, N\}$ and $p \in \P_i$, $\alpha^{i,p}_{a,b} : \Z_i \to \reals$ such that
$\alpha^{i,p}_{a,b}(z)$ gives the (constant) rate (i.e., the number of vehicles per time slot) at which vehicles can go from $\R_a$ to $\R_b$ through junction $\J_i$ when the traffic state is $z$ and phase $p$ is activated.
For a time slot $t \in \naturals$, we define the ``pressure relief''  associated phase $p \in \P_i$ at junction $\J_i, i \in \{1, \ldots, L\}$ by
$\pr_i^t(p) \triangleq \sum_{(\R_a,\R_b) \in p} W_{a,b}(t) \alpha^{i,p}_{a,b}(z(t))$.

Consider an arbitrary junction $\J_i, i \in \{1, \ldots, L\}$.
For the unconstrained case where $\underline{T}^i_p = 0$ and $\overline{T}^i_p = 1$ for all $p \in \P_i$,
it can be checked that a solution $S^*$ of (\ref{eq:split_comp}) agrees with that presented in \cite{Wongpiromsarn:ITSC2012}.
In this case, $S^*$ is given by
\begin{equation*}
  S^*(p) = \left\{ \begin{array}{ll} 
    1 &\hbox{ if } p = p^*,\\
    0 &\hbox{ otherwise},
  \end{array} \right.
\end{equation*}
where $p^*$ is a phase satisfying
$\pr_i^t(p^*) \geq \pr_i^t(p)$
for all $p \in \P_i$.
That is, in each time slot, the controller activates only one phase with the maximum associated pressure relief.

For a more general case where $\underline{T}^i_p, \overline{T}^i_p \in [0, 1]$ for all $p \in \P_i$,
we compute an order
$p_1, p_2, \ldots, p_{r_i}$ where $p_j \in \P_i$ for all $j \in \{1, \ldots, r_i\}$ and $r_i$ is the cardinality of $\P_i$
such that
$\pr_i^t(p_j) \geq \pr_i^t(p_k)$ for all $k > j$.
A solution $S^*$ of (\ref{eq:split_comp}) is then given by
\begin{eqnarray*}
S^*(p_1) &=& \min(\overline{T}^i_{p_1}, 1 - \sum_{j=2}^{r_i} \underline{T}^i_{p_j}),\\
S^*(p_k) &=& \min(\overline{T}^i_{p_k}, 1 - \sum_{j=1}^{k-1} S^*(p_j) - \sum_{j=k+1}^{r_i} \underline{T}^i_{p_j}),
\qquad \forall k \in \{2, \ldots, r_i\}.
\end{eqnarray*}
Note that since the number of possible phases for each junction is typically small (e.g., less than 10), the above computation
and ordering of phases $p_1, p_2, \ldots, p_{r_i} \in \P_i$
can be practically performed in real time.

%%%%%%%%%%%%%%%%%%%%%%%%%%%%%%%%%%%%%%%%%%%%%%%%%%%%%%%%%%%%%%%
\section{Simulation Results}
\label{sec:results}
In this section, we employ a microscopic traffic simulator MITSIMLab \cite{Ben-Akiva01}, 
whose simulation models have been validated against traffic data collected from Swedish cities,
to evaluate the performance of our backpressure-based traffic signal control algorithm
in comparison with the SCATS-like algorithm as explained in \cite{Liu03Master}.
A medium size road network of the Marina Bay area of Singapore
with 112 links and 14 signalized junctions as shown in Figure \ref{fig:sim-MITSIM-network} is considered.
We implement the SCATS-like and our backpressure-based traffic signal control algorithms 
in the traffic management simulator component of MITSIMLab.
For the SCATS-like implementation, the number of possible split plans for each junction ranges from 5 to 17.
The standard space time under saturated flow for each vehicle is assumed to be 0.96 seconds.
The maximum, minimum and medium cycle lengths are set to 140 seconds, 60 seconds and 100 seconds, respectively.
The degrees of saturation that would result in the maximum, minimum and medium cycle lengths are assumed to be 0.9, 0.3 and 0.5, respectively. Finally, the split plan is computed based on the vote from the last 5 cycles.

Two implementations of the backpressure-based algorithms are used in the evaluation.
The first implementation, denoted by UBP, represents the unconstrained case where 
$\underline{T}^i_p = 0$ and $\overline{T}^i_p = 1$ for all $p \in \P_i$, $i \in \{1, \ldots, 14\}$
whereas
the second implementation, denoted by CBP, represents the constrained case with 
$\underline{T}^i_p = 0.15$ and $\overline{T}^i_p = 0.7$ for all $p \in \P_i$, $i \in \{1, \ldots, 14\}$.
In both implementations, we assume that the flow rate through a junction for each phase during each time slot is constant
so that the function $\xi^{i}_{a,b}$ can be written as in (\ref{eq:xi_const_flow}) where
the constant $\alpha^{i,p}_{a,b}(z)$ is obtained from the corresponding flow rate in the previous time slot.

Using two case studies, the performance of both algorithms is evaluated based on different measures, 
including queue length, delay and number of stops.
In the first case study, the origin-destination pairs are calibrated such that the traffic
resulting from applying the SCATS-like algorithm closely matches the real situation 
(for which the parameters of the SCATS-like algorithm were calibrated).
In the second case study, we perturb the calibrated origin-destination pairs
to illustrate the robustness of our backpressure-based traffic signal control algorithm.  

\begin{figure}
   \centering 
        \includegraphics[trim=5.5cm 2.7cm 1.5cm 1.6cm, clip=true, width=0.6\textwidth]{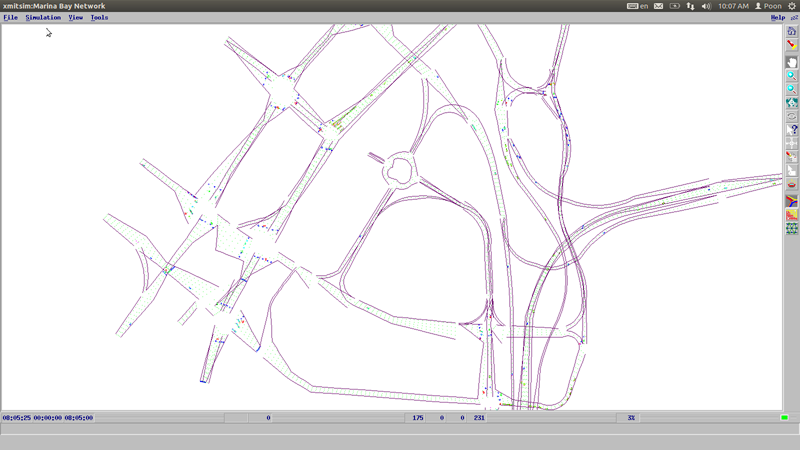}
   \caption{Road network used in the MITSIMLab simulation.}
  \label{fig:sim-MITSIM-network}
\end{figure}

\subsection{Calibrated Origin-Destination Pairs}
In this case study, the origin-destination pairs have been calibrated such that the traffic
resulting from applying the SCATS-like algorithm closely matches the real situation.
Vehicles exogenously enter and exit the network at various links based on 61 different origin-destination pairs.
The vehicle arrival rate varies with time and ranges from 10543 vehicles/hour to 13341 vehicles/hour.
The maximum and average queue lengths are shown in Figure \ref{fig:sim-results-MITSIM1}, 
illustrating that the performance of our algorithm, both in the unconstrained and constrained cases,
is comparable to that of the SCATS-like algorithm.
Overall, the maximum queue lengths for the SCATS-like algorithm, UBP and CBP are
129, 119 and 98, respectively,
whereas the average queue lengths for the SCATS-like algorithm, UBP and CBP are approximately
7.4, 6.1 and 5.4, respectively.

\begin{figure}
   \centering 
        \includegraphics[trim=3cm 0cm 3cm 1cm, clip=true, width=0.49\textwidth]{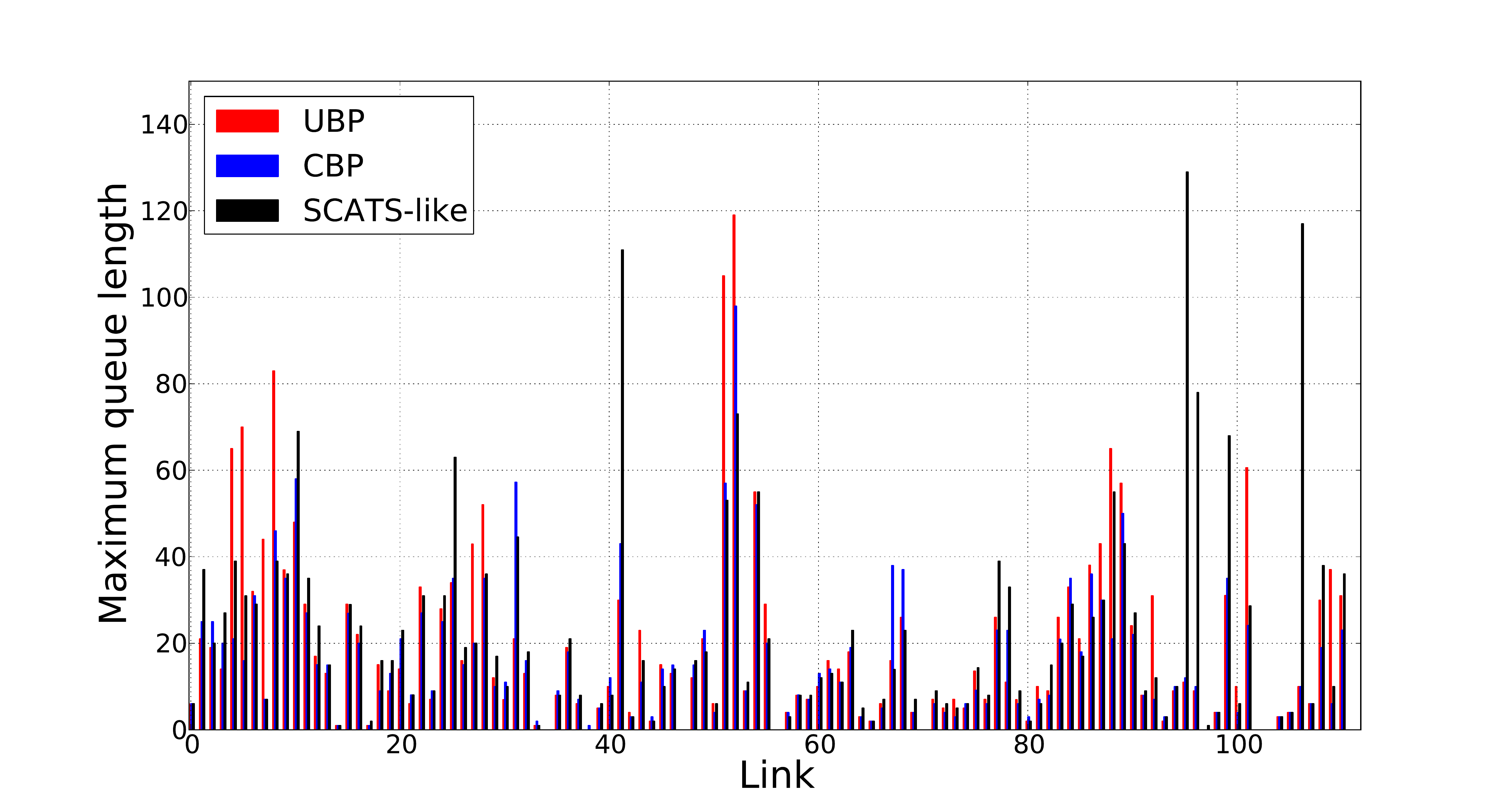}
        \includegraphics[trim=3cm 0cm 3cm 1cm, clip=true, width=0.49\textwidth]{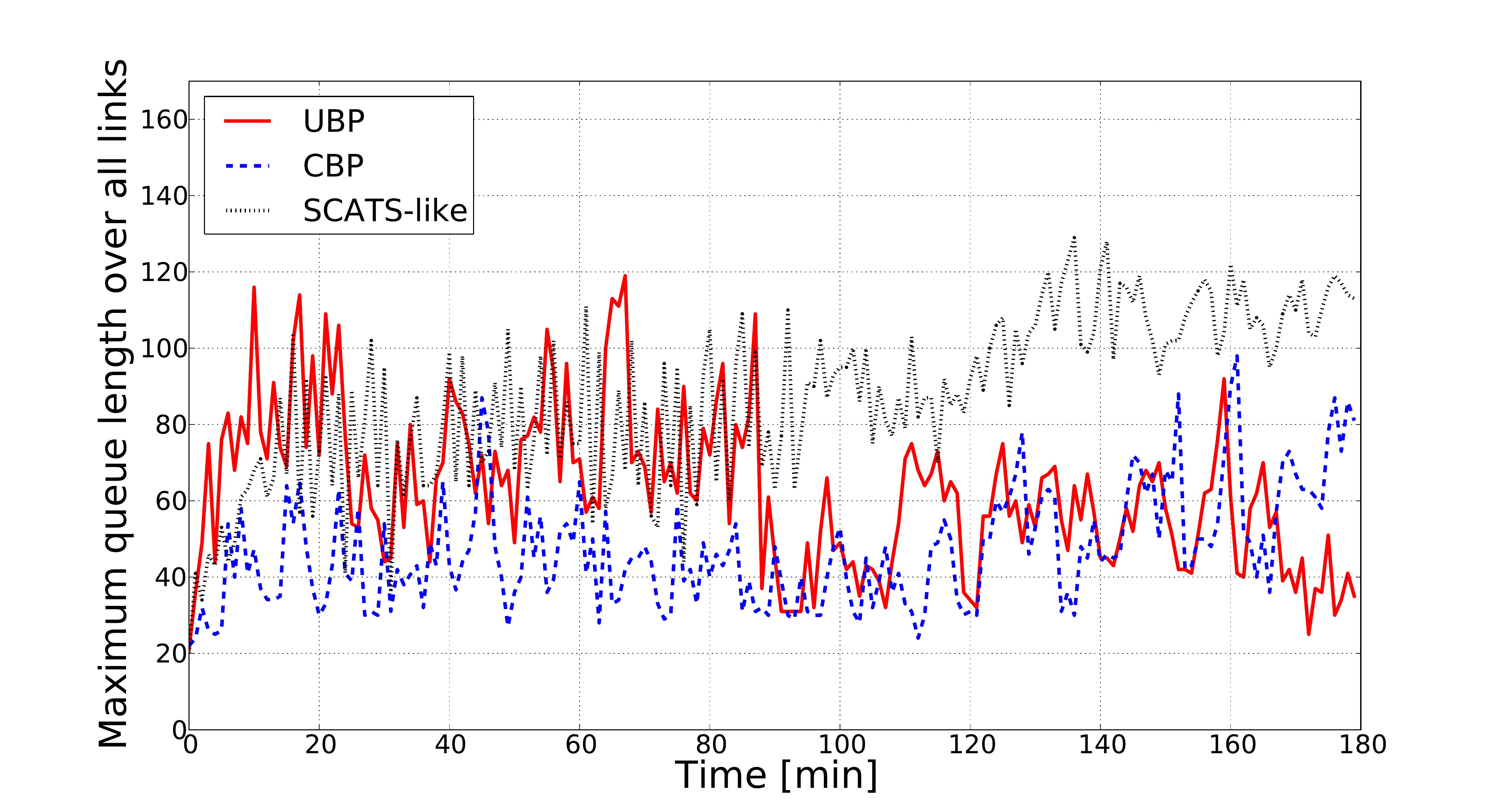}
        \includegraphics[trim=3cm 0cm 3cm 1cm, clip=true, width=0.49\textwidth]{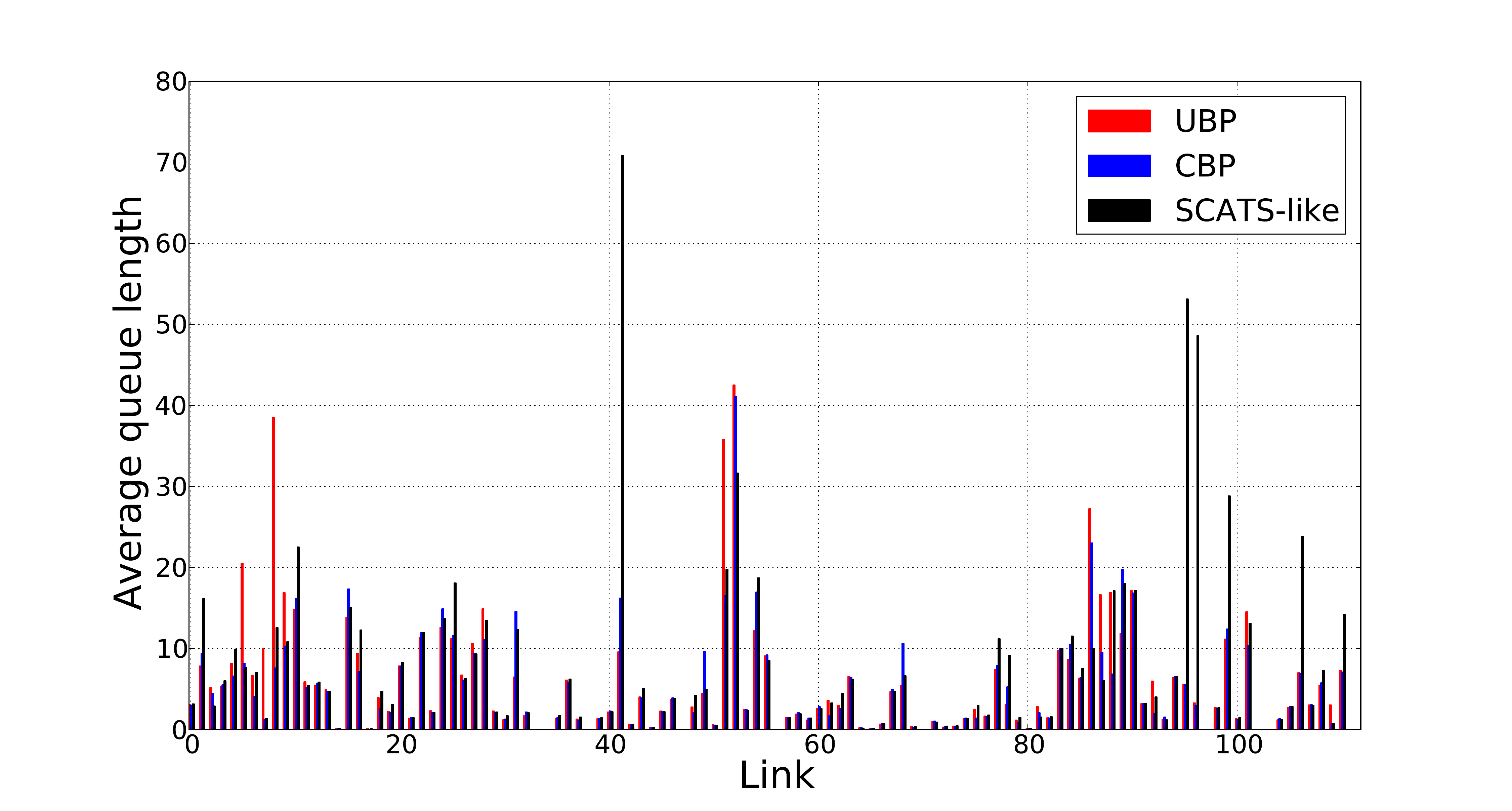}
        \includegraphics[trim=3cm 0cm 3cm 1cm, clip=true, width=0.49\textwidth]{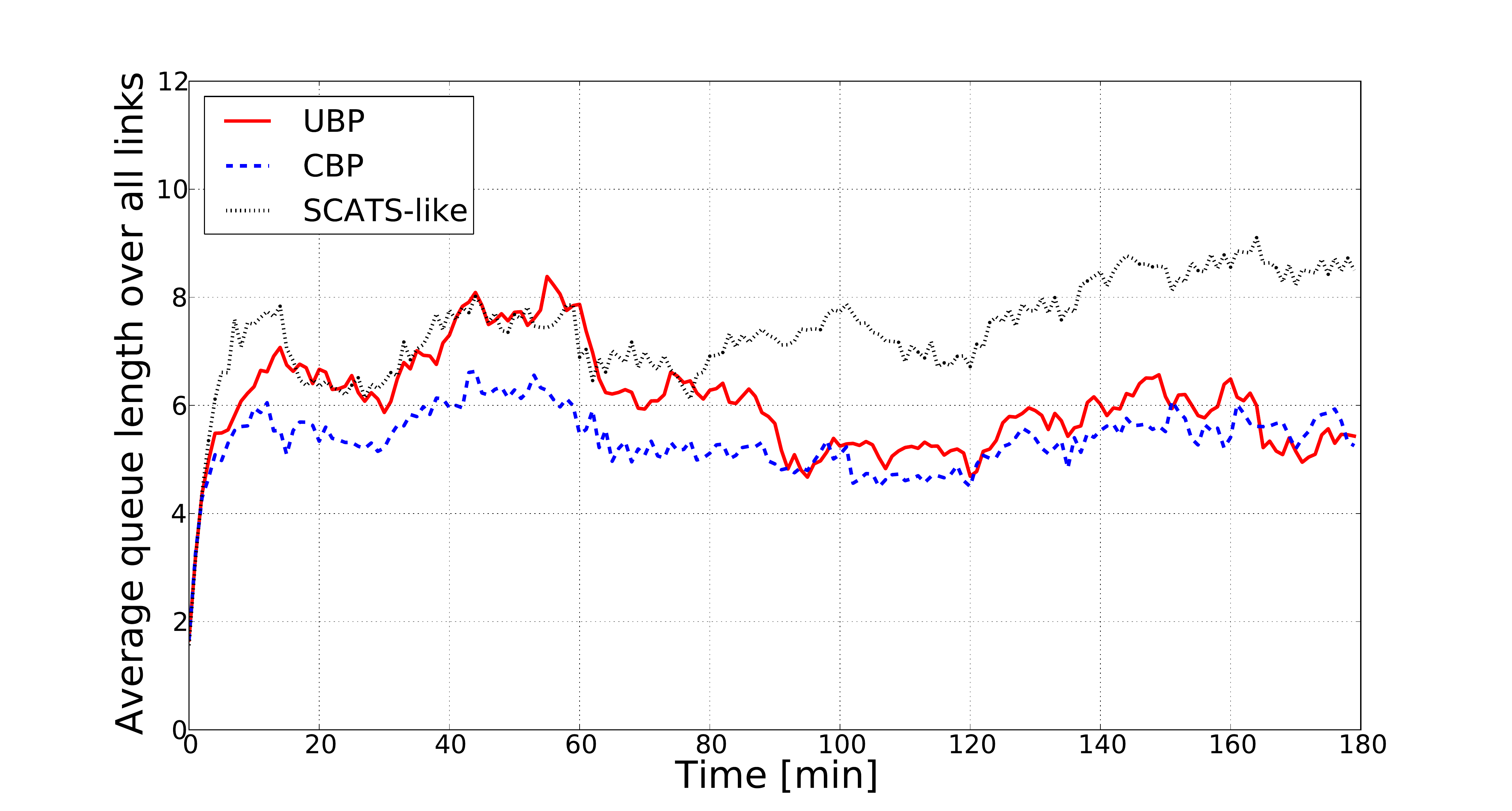}
   \caption{Simulation results showing maximum and average queue lengths when the SCATS-like and our backpressure-based traffic signal control algorithm are used with the calibrated origin-destination pairs.}
  \label{fig:sim-results-MITSIM1}
\end{figure}

The average and maximum delays for each origin-destination pair are shown in Figure \ref{fig:sim-MITSIM-delay1}.
When the SCATS-like algorithm, UBP and CBP are applied,
the average delays over all the vehicles are computed to be approximately 253,  249 and 202 seconds, respectively,
whereas the maximum delays are  2324, 2387 and 1324 seconds, respectively.
Finally, the average number of stops per vehicle on each link is shown in Figure \ref{fig:sim-MITSIM-stops1}.
The average numbers of stops per vehicle when the SCATS-like algorithm, UBP and CBP are applied are approximately 
3.6, 3.1 and 2.5 respectively.

\begin{figure}
   \centering 
        \includegraphics[trim=2.5cm 0cm 3cm 1cm, clip=true, width=0.49\textwidth]{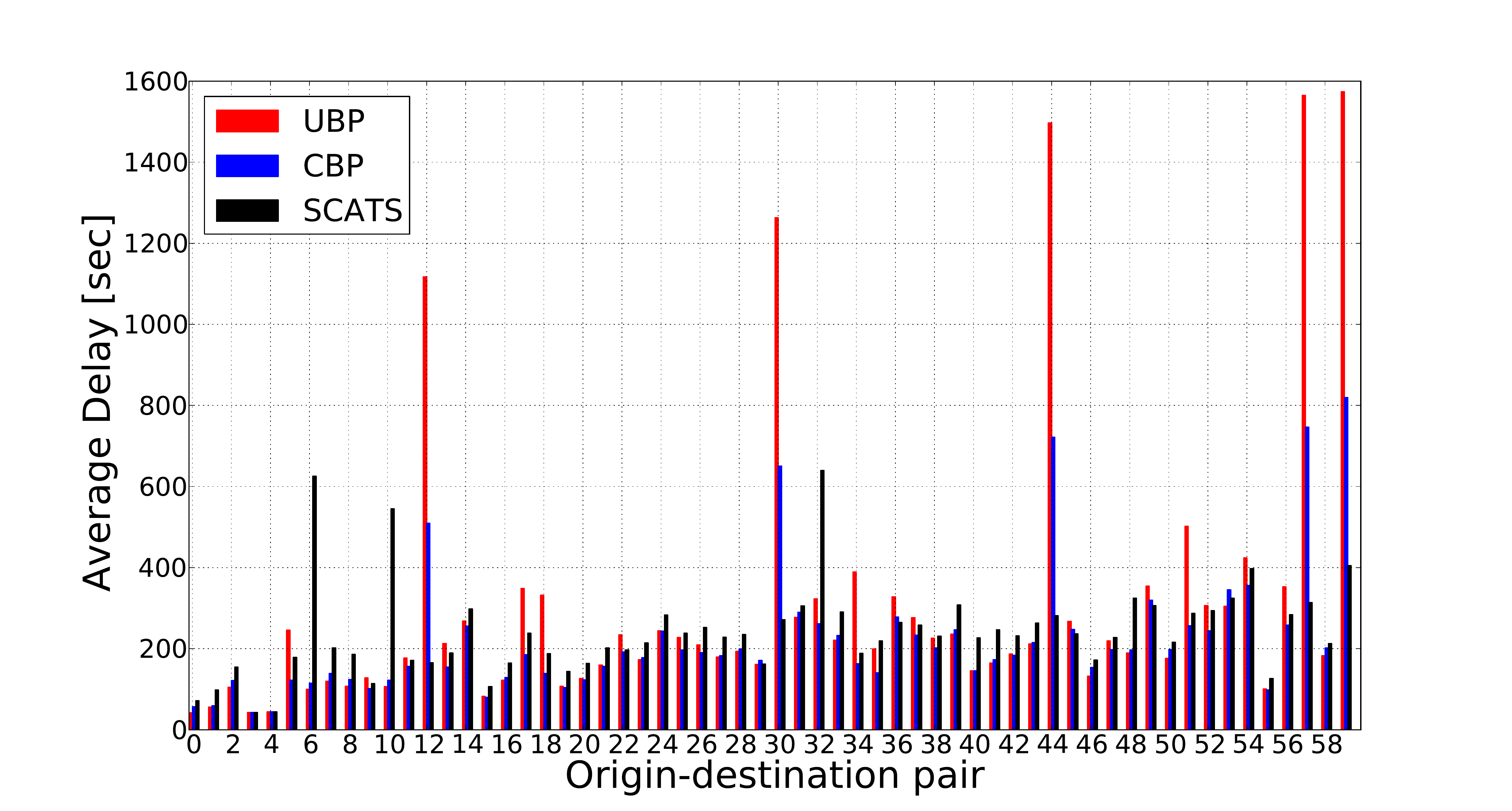}
        \includegraphics[trim=2.5cm 0cm 3cm 1cm, clip=true, width=0.49\textwidth]{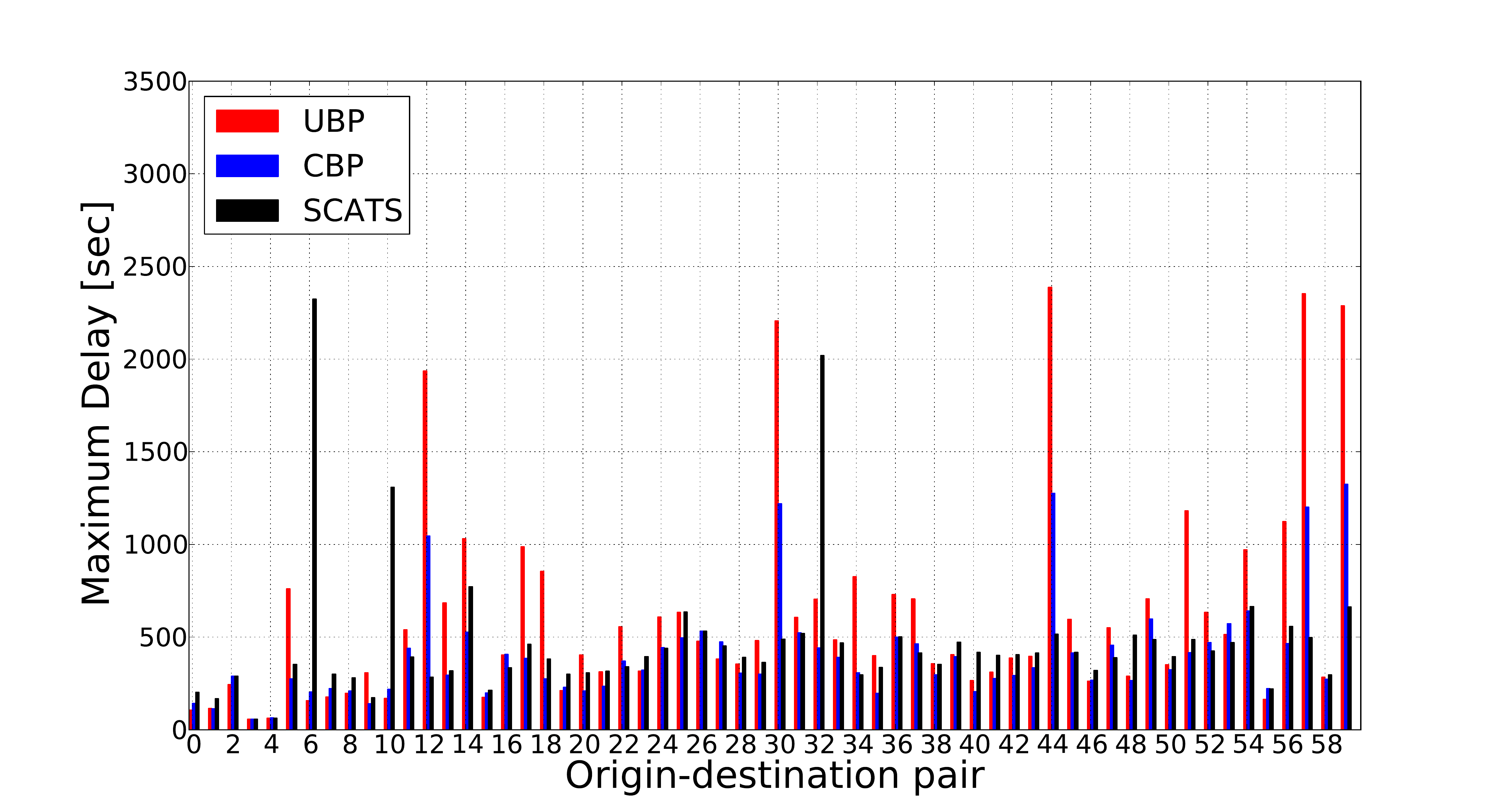}
   \caption{Simulation results showing 
   (left) average delay and 
   (right) maximum delay
   for each calibrated origin-destination pair when the SCATS-like and our backpressure-based traffic signal control algorithm are used.}
  \label{fig:sim-MITSIM-delay1}
\end{figure}

\begin{figure}
   \centering 
        \includegraphics[trim=3cm 0cm 3cm 1cm, clip=true, width=0.49\textwidth]{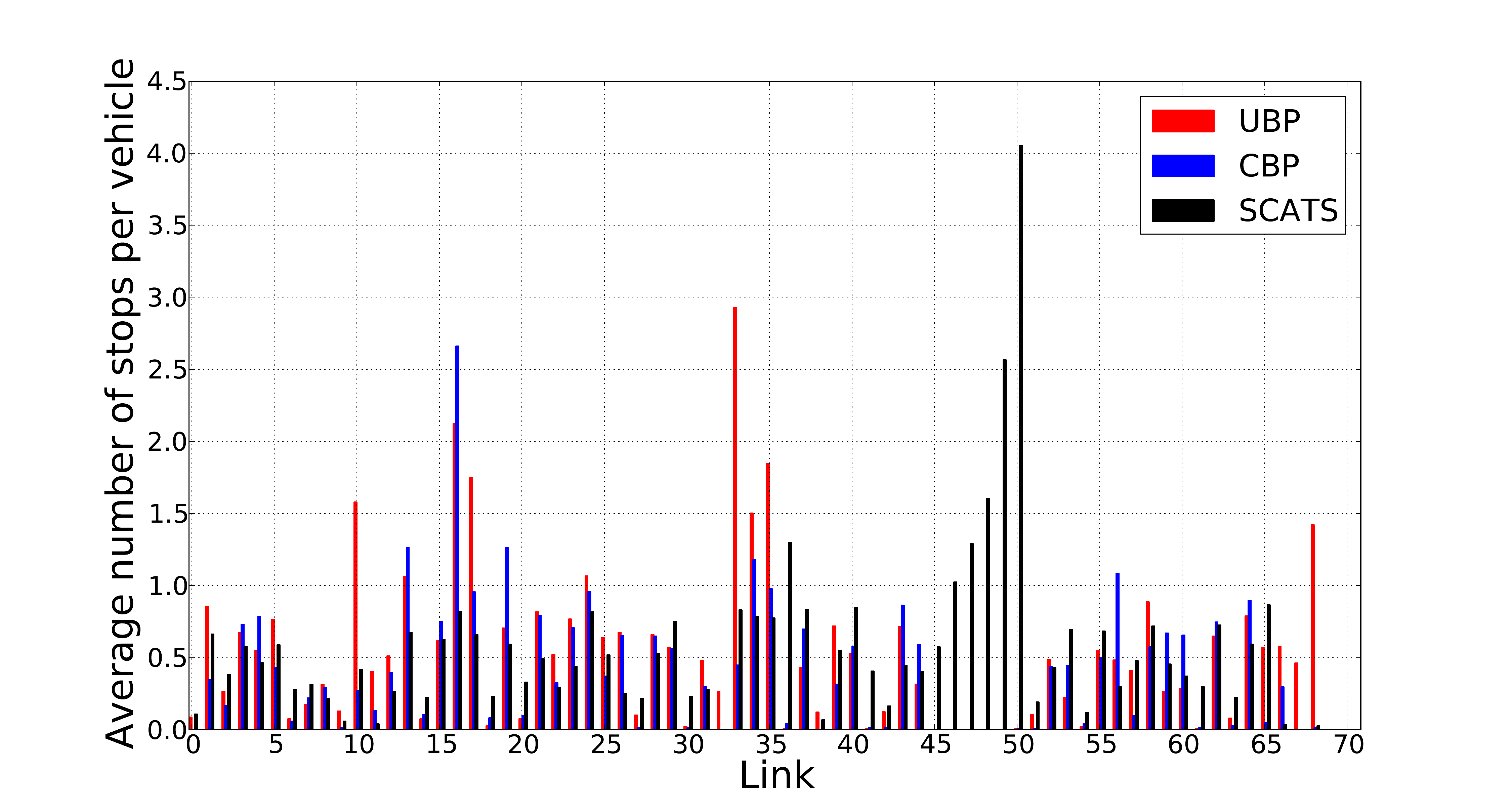}
   \caption{Simulation results showing 
   average number of stops per vehicle on each link 
   when the SCATS-like and our backpressure-based traffic signal control algorithm are used
   with the calibrated origin-destination pairs.
   Note that only links with a nonzero number of stops when both SCATS and our algorithm are applied are shown.}
  \label{fig:sim-MITSIM-stops1}
\end{figure}

The unexpected superior performance of CBP over UBP potentially results from the inaccurate estimate of the function
$\xi^{i}_{a,b}$ and the inaccurate queue length measurement in the simulation.
In the implementation of the backpressure algorithm in MITSIMLab, queue length on each link is obtained by subtracting
the count of the vehicles leaving the end of the link from the count of the vehicles entering the beginning of the link, thus failing
to account for vehicles that change lane (i.e., those that leave or enter the link anywhere besides either ends of the link).
This results in an inaccurate computation of the pressure relief associated with each phase.
Since in each time slot, CBP gives the right of way to all the phases, whereas
UBP only gives the right of way to a single phase,
CBP is expected to be more robust to such an inaccurate pressure relief computation
and an inaccurate estimate of the function $\xi^{i}_{a,b}$.

\subsection{Perturbed Origin-Destination Pairs}
In this case study, vehicles exogenously enter and exit the network at various links based on 46 different origin-destination pairs,
with the arrival rate of 9330 vehicles/hour.
The simulation video can be found at \url{http://youtu.be/Sk-d5-cfkDk}.
The maximum and average queue lengths are shown in Figure \ref{fig:sim-results-MITSIM2}.
Overall, the maximum queue lengths for the SCATS-like algorithm, UBP and CBP are
305, 56 and 69, respectively,
whereas the average queue lengths for the SCATS-like algorithm, UBP and CBP are approximately
8.8, 3.1 and 2.9, respectively.
These simulation results show that our algorithm
(both UBP and CBP)
can significantly reduce the maximum and average queue lengths compared to the SCATS-like algorithm.
%A long queue that causes spillback over several links in the network is observed when SCATS is used.
In addition, as shown in Figure \ref{fig:sim-MITSIM-spillback}, 
queue spillback, where queues extend beyond one link upstream from the junction, persists throughout the simulation,
especially when the SCATS-like controller is used.
%At this high arrival rate, queue spillback, which persists throughout the simulation, is observed when both SCATS and our algorithm are used.
%However, the situation is much worse when SCATS is used as shown in Figure \ref{fig:sim-MITSIM-spillback}.

The average and maximum delays for each origin-destination pair are shown in Figure \ref{fig:sim-MITSIM-delay2}.
When the SCATS-like algorithm, UBP and CBP are applied,
the average delays over all the vehicles are computed to be approximately 277,  172 and 123 seconds, respectively,
whereas the maximum delays are  7954, 2430 and 558 seconds, respectively.
%These simulation results show that our algorithm can reduce the average and maximum delay by approximately
%38\% and 69\%, respectively, compared to SCATS.

\begin{figure}
   \centering 
        \includegraphics[trim=3cm 0cm 3cm 1cm, clip=true, width=0.49\textwidth]{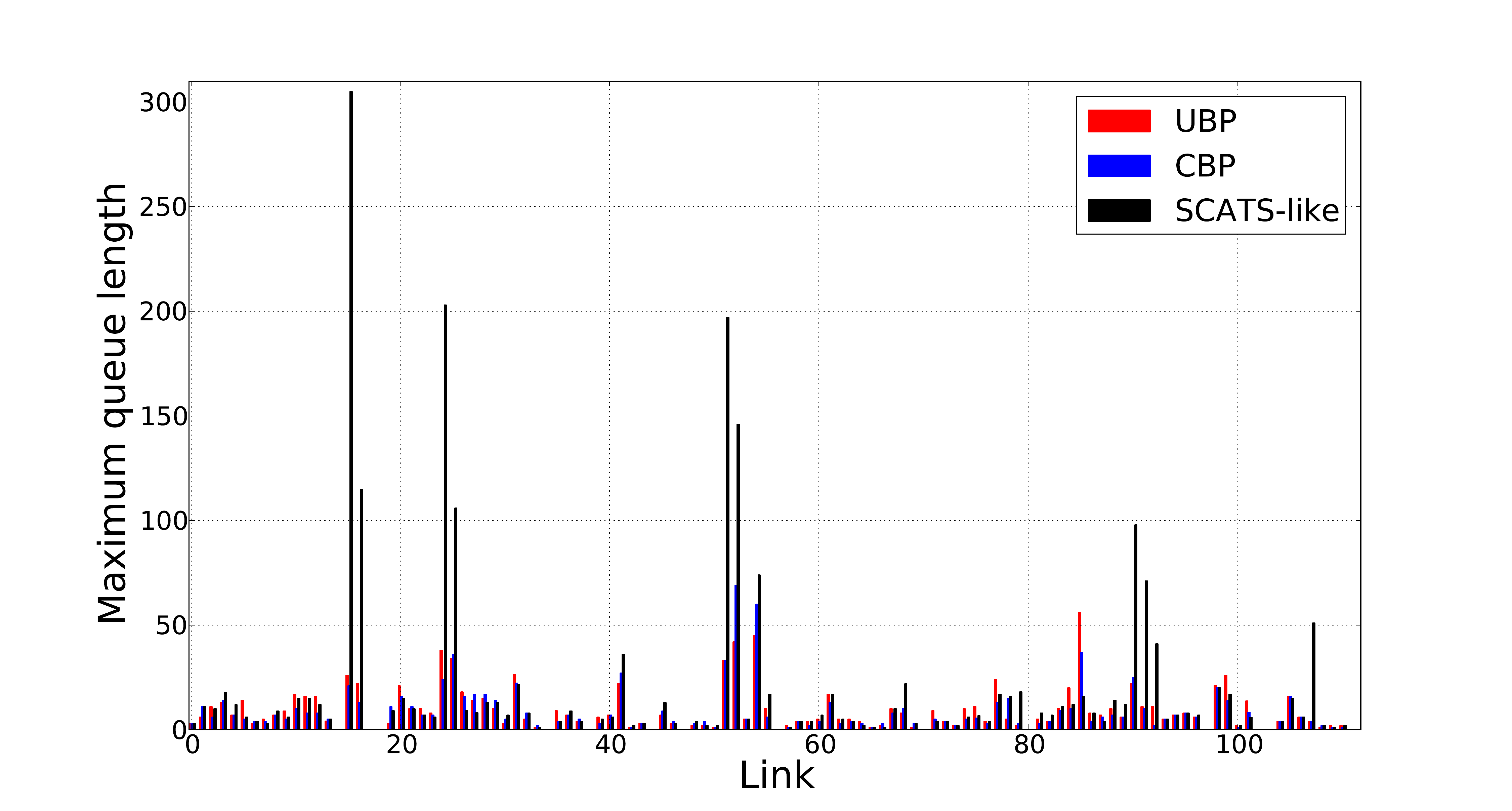}
        \includegraphics[trim=3cm 0cm 3cm 1cm, clip=true, width=0.49\textwidth]{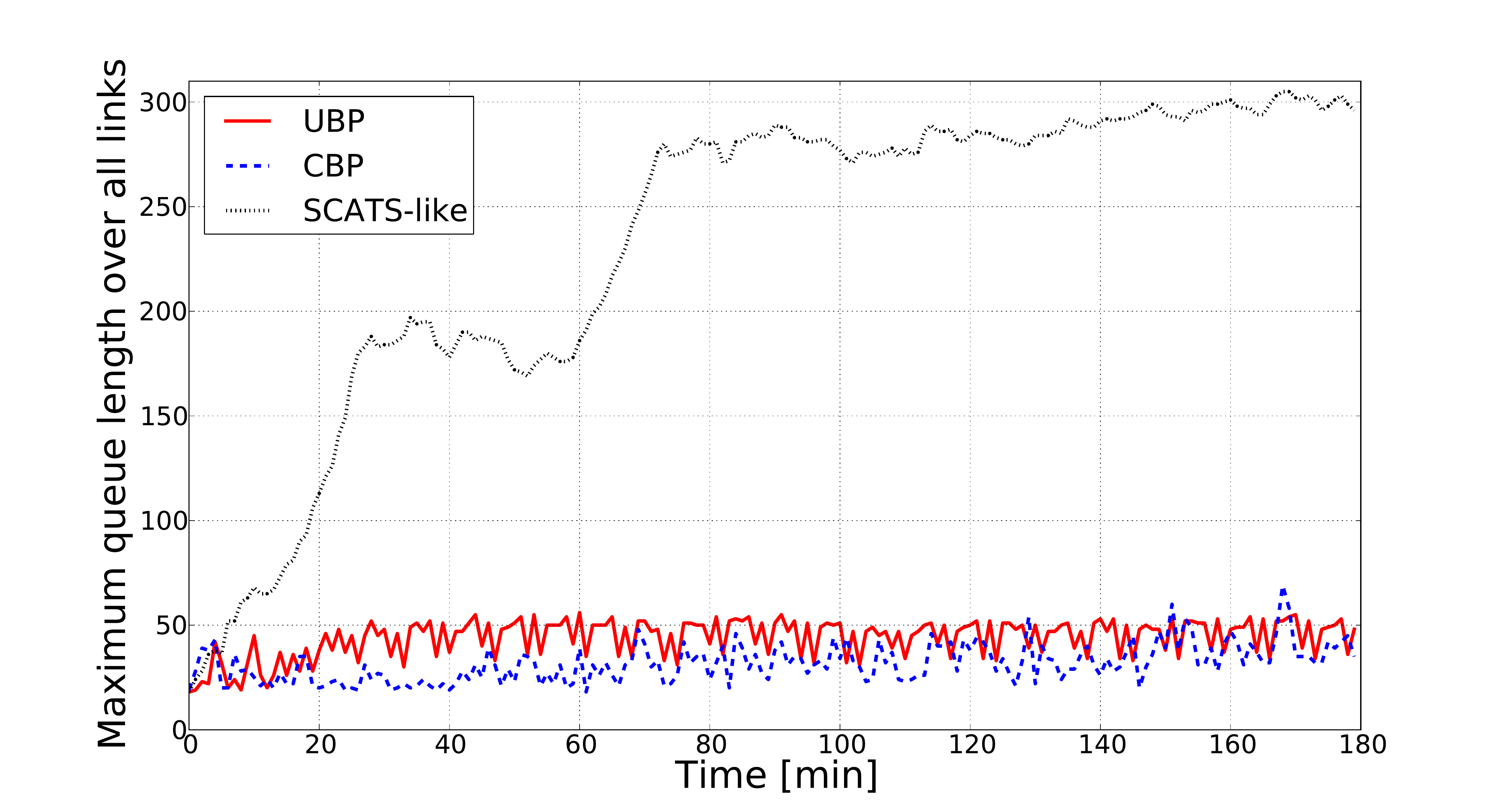}
        \includegraphics[trim=3cm 0cm 3cm 1cm, clip=true, width=0.49\textwidth]{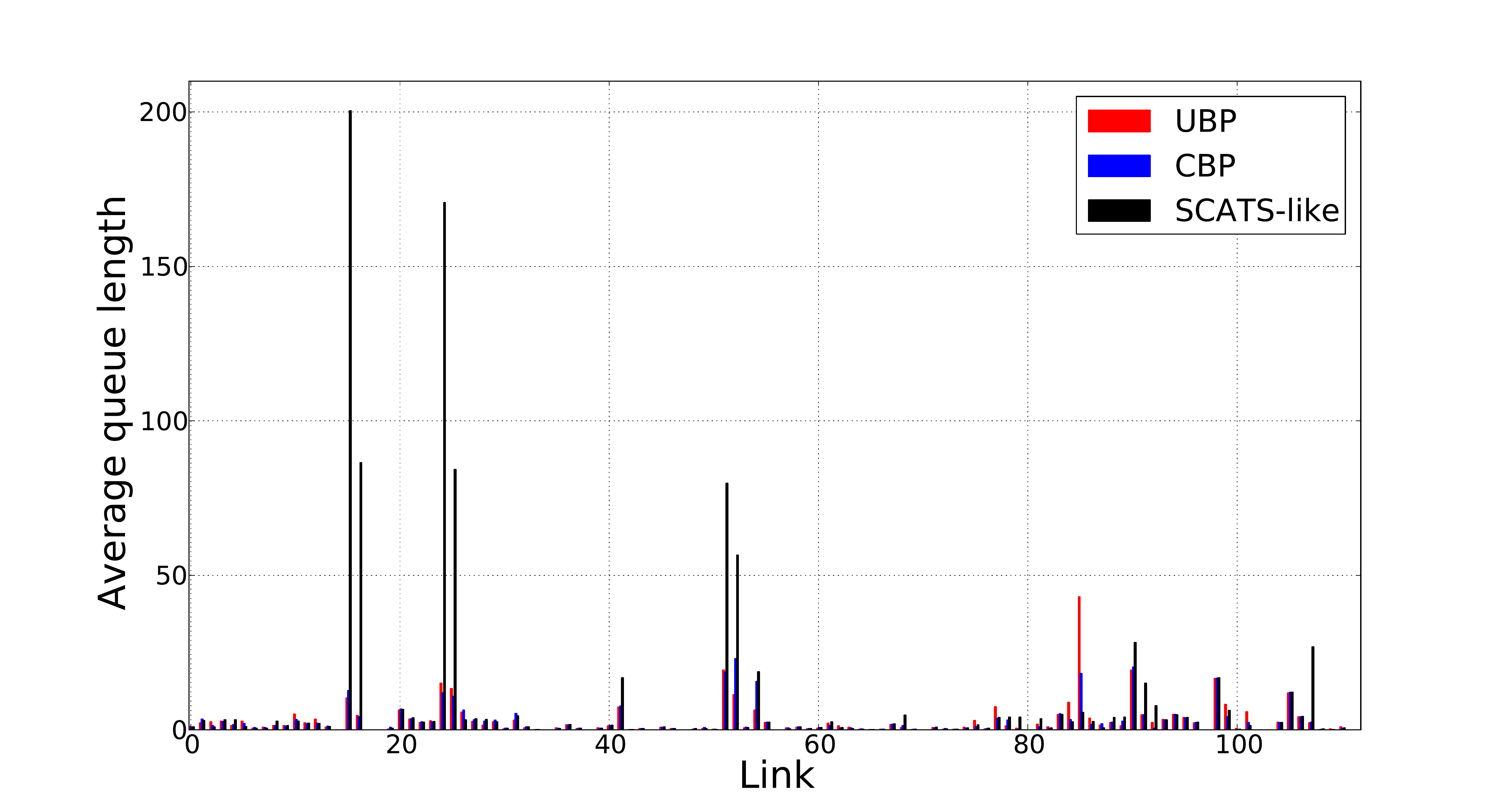}
        \includegraphics[trim=3cm 0cm 3cm 1cm, clip=true, width=0.49\textwidth]{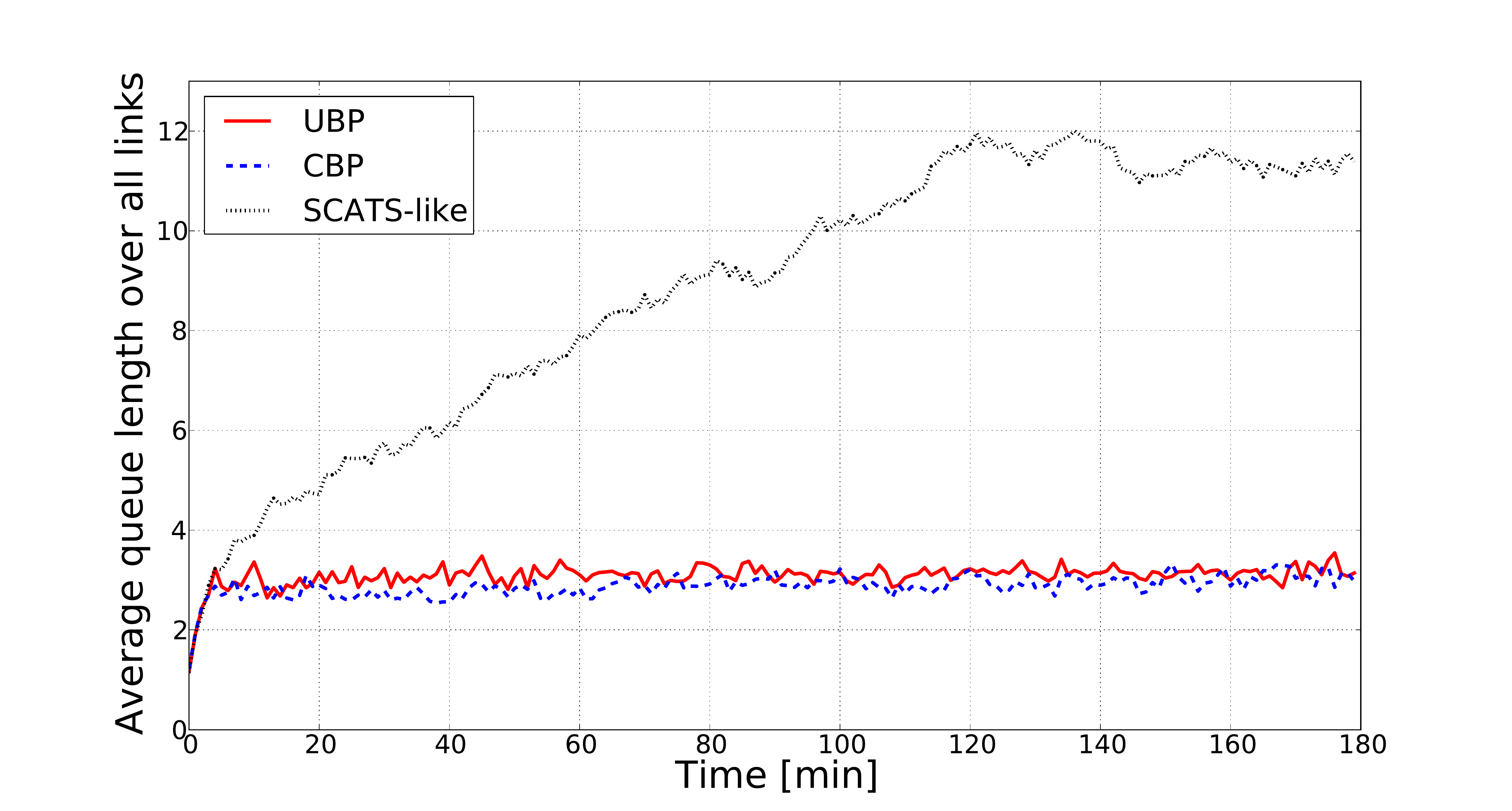}
   \caption{Simulation results showing maximum and average queue lengths when the SCATS-like and our backpressure-based traffic signal control algorithm are used with the perturbed origin-destination pairs.}
  \label{fig:sim-results-MITSIM2}
\end{figure}

\begin{figure}
   \centering 
        \includegraphics[trim=16.5cm 2.7cm 20.5cm 1.6cm, clip=true, width=0.3\textwidth]{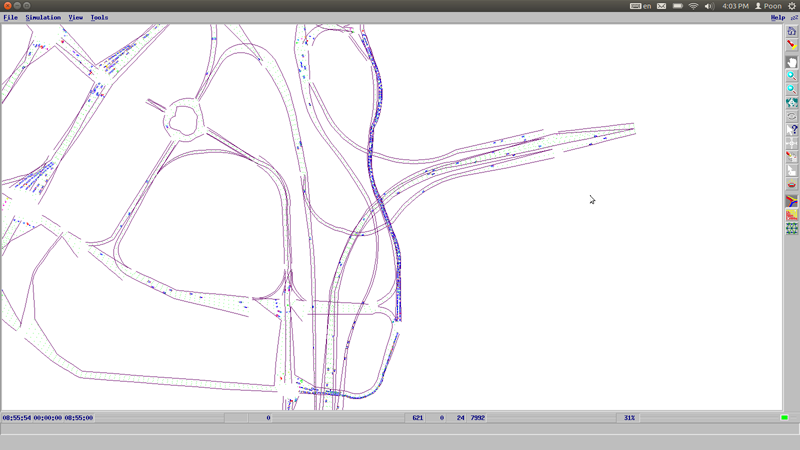}
        \hspace{30mm}
        \includegraphics[trim=16.5cm 2.7cm 20.5cm 1.6cm, clip=true, width=0.3\textwidth]{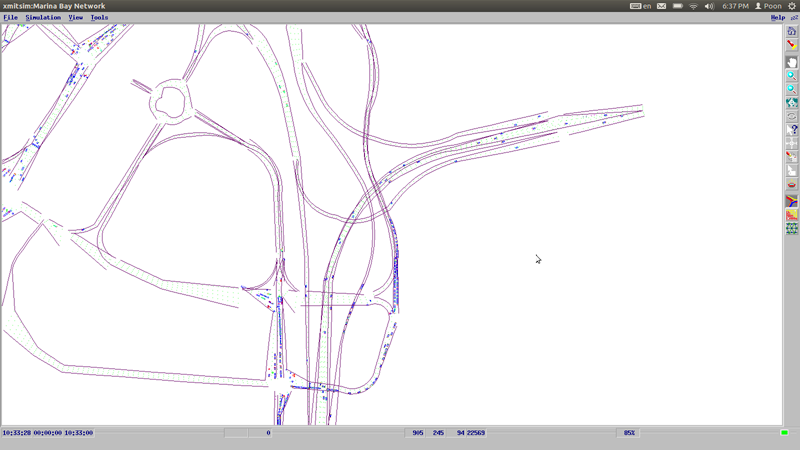}
   \caption{%(left) Queue spillback which occurs when SCATS is used, and 
   (left) Queues spread over multiple links upstream when the SCATS-like algorithm is used, and
   (right) Queues do not spread over as many links when our backpressure-based traffic signal control algorithm is used.
   The part of the road that is filled with blue is occupied by vehicles.}
  \label{fig:sim-MITSIM-spillback}
\end{figure}

\begin{figure}
   \centering 
        \includegraphics[trim=2.5cm 0cm 3cm 1cm, clip=true, width=0.49\textwidth]{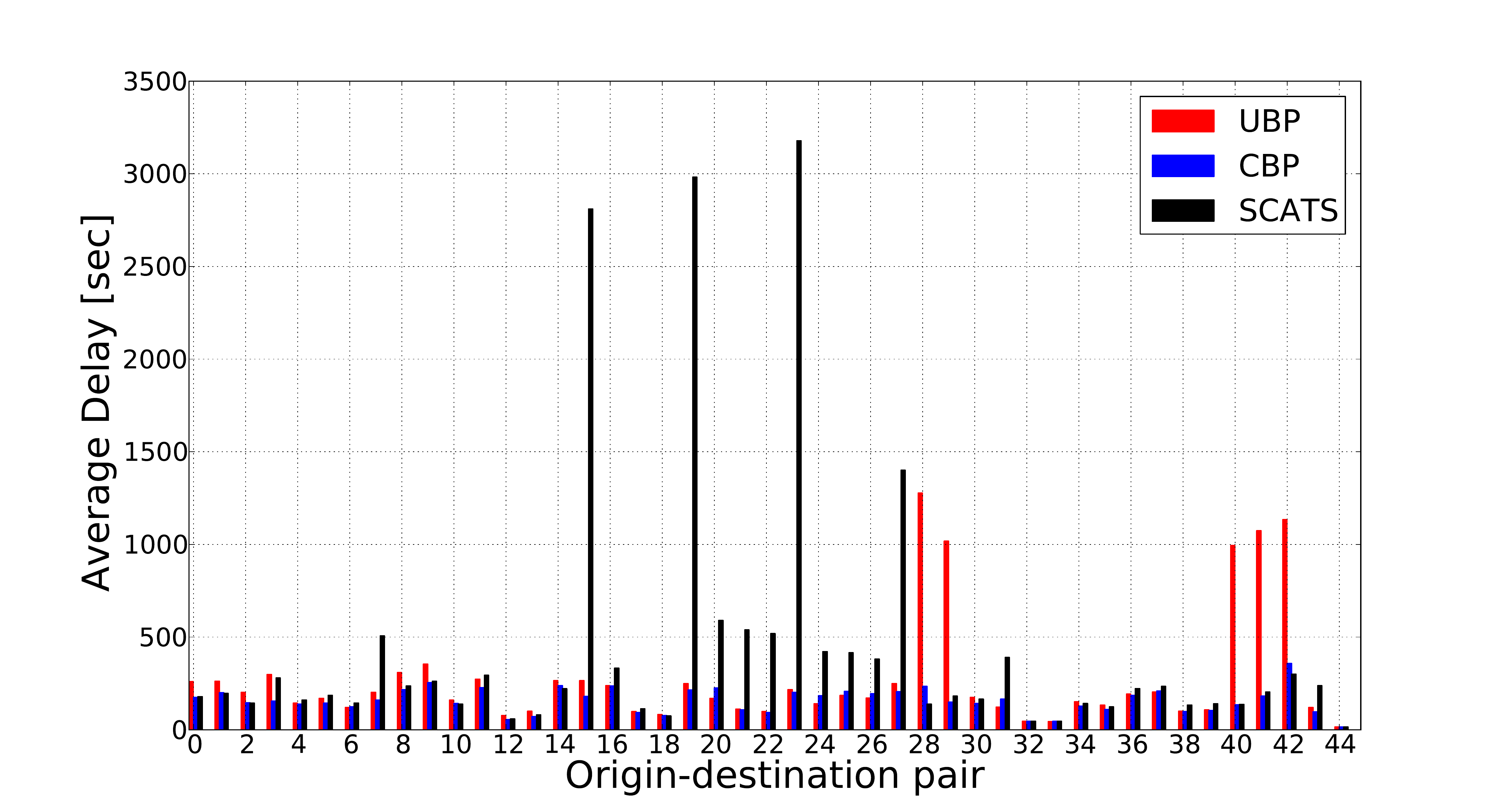}
        \includegraphics[trim=2.5cm 0cm 3cm 1cm, clip=true, width=0.49\textwidth]{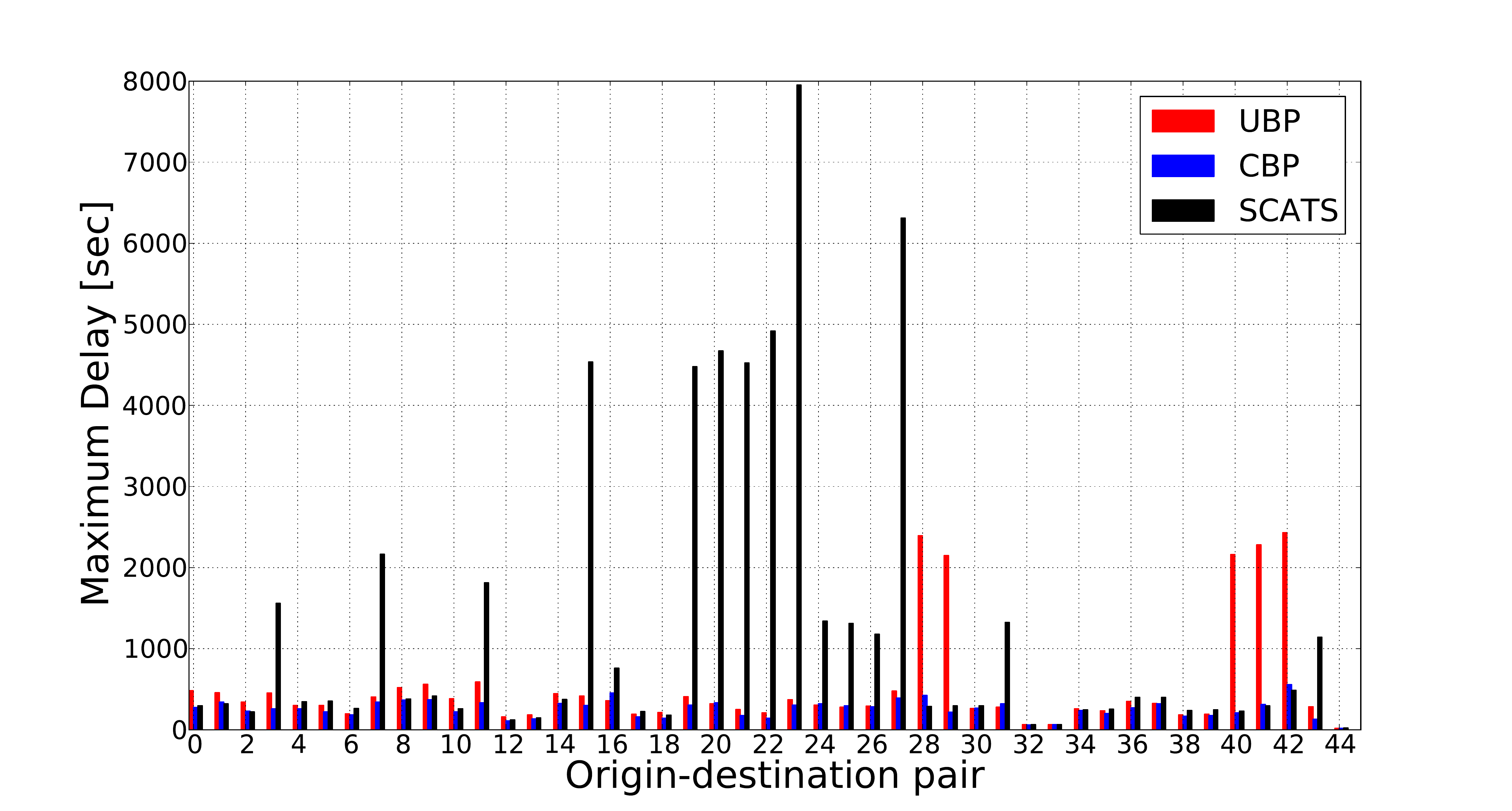}
   \caption{Simulation results showing 
   (left) average delay and 
   (right) maximum delay
   for each perturbed origin-destination pair when the SCATS-like and our backpressure-based traffic signal control algorithm are used.}
  \label{fig:sim-MITSIM-delay2}
\end{figure}

Finally, the average number of stops per vehicle on each link is shown in Figure \ref{fig:sim-MITSIM-stops2}.
The average numbers of stops per vehicle when the SCATS-like algorithm, UBP and CBP are applied are approximately 
7, 1 and 1, respectively.
This shows that even though our algorithm is completely distributed and does not explicitly
enforce the coordination among the traffic light controllers at neighboring junctions, a green wave is still achieved.

\begin{figure}
   \centering 
        \includegraphics[trim=3cm 0cm 3cm 1cm, clip=true, width=0.49\textwidth]{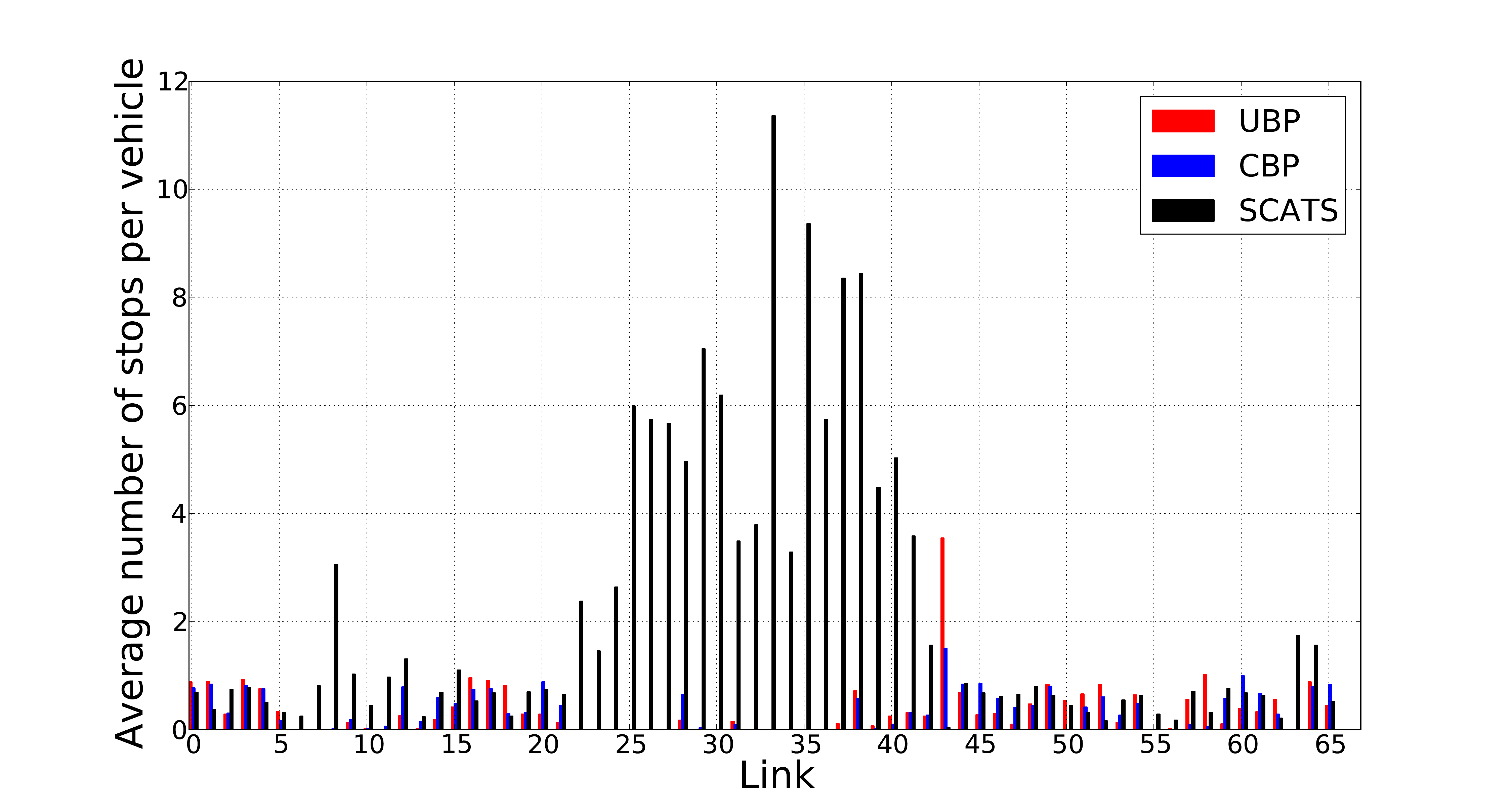}
   \caption{Simulation results showing 
   average number of stops per vehicle on each link 
   when the SCATS-like and our backpressure-based traffic signal control algorithm are used
   with the perturbed origin-destination pairs.
   Note that only links with a nonzero number of stops when both SCATS and our algorithm are applied are shown.}
  \label{fig:sim-MITSIM-stops2}
\end{figure}

%%%%%%%%%%%%%%%%%%%%%%%%%%%%%%%%%%%%%%%%%
\section{Discussion}
\label{sec:discussion}
The backpressure-based traffic signal control algorithm is a computationally simple and robust 
method that leads to maximum network throughput.
The algorithm is completely distributed, i.e., the signal at each junction is determined completely independently 
from other junctions.
As a result, it can be applied to an arbitrarily large network. 
Besides offering superior network performance based on standard measures such as queue length,
delay and number of stops, key advantages over existing algorithms include:
\begin{enumerate}
\item Ease of implementation: 
As opposed to SCATS where 
each junction needs to be identified as critical or non-critical and
all the possible split plans need to be pre-specified and tuned based on the characteristics of the traffic on the network,
the backpressure-based traffic signal control algorithm treats all the junctions exactly the same
and does not require a pre-defined set of all the possible split plans.

\item Robustness:
As the backpressure-based traffic signal control algorithm does not rely on a pre-defined set of split plans
and  an identification of critical junctions,
it is more robust to changes in the characteristics of the traffic and the network, including
changes in the origin-destination pairs 
(e.g., when a new structure is introduced to the network or an important event occurs),
and changes in the road conditions.

\item Computational simplicity: 
As opposed to existing optimization-based techniques where a large optimization problem needs to be solved,
considering the complete network,
the backpressure-based traffic signal control algorithm solves an optimization problem for each individual junction separately.
Hence, the size of the problem is independent of the size of the road network.
Furthermore, as discussed in Section \ref{sec:example},
for the special case where the flow rate through a junction for each phase is constant and only depends on the traffic state, 
the backpressure-based traffic signal control algorithm only requires a simple algebraic computation,
using only local information.

\end{enumerate}

%%%%%%%%%%%%%%%%%%%%%%%%%%%%%%%%%%%%
\section{Conclusions and Future Work}
\label{sec:conclusions}
We considered distributed control of traffic signals.
Motivated by backpressure routing, which has been mainly applied to communication and power networks,
our approach relies on constructing a set of local controllers, each of which is associated with each junction.
These local controllers are constructed and implemented independently of one another.
Furthermore, each local controller does not require the global view of the road network.
Instead, it only requires information that is local to the junction with which it is associated.
Constraints such as periodic switching sequences of phases and minimum and maximum green time can be incorporated.
We formally proved that our algorithm leads to maximum network throughput even though
the controller is constructed and implemented in such a distributed manner
and no information about traffic arrival rates is provided.
Simulation results illustrate the effectiveness of our algorithm.
%Simulation results showed that our algorithm performs significantly better than the existing SCATS-like algorithm.
%an adaptive traffic signal control systems that is being used in many cities.

Future work includes taking into account road capacity.
We are also considering the coordination issue such as ensuring the emergence of green waves.
Finally, we are investigating methods for measuring the ``pressure relief'' without having to find the queue length on each link.

%%%%%%%%%%%%%%%%%%%%%%%%%%%%%%%%%%%%%%%%%%%%%%%%%%%%%%%%%%%%%%%%%%%%%%%%%%%%%%%%
\section{ACKNOWLEDGMENTS}
The authors gratefully acknowledge Ketan Savla for the inspiring discussions, Prof. Moshe Ben-Akiva and his research group, in particular
Kakali Basak and Linbo Luo, for support with MITSIMLab,
and Land Transport Authority of Singapore for providing the data collected from the loop detectors
installed in the Marina Bay area of Singapore.
This work is supported in whole or in part by the Singapore National Research Foundation (NRF) through the Singapore-MIT Alliance for Research and Technology (SMART) Center for Future Urban Mobility (FM).

\bibliographystyle{IEEEtran}
\bibliography{IEEEabrv,ref}
\end{document}